\documentclass[11pt]{article}
\usepackage[margin=1.1in]{geometry}
\usepackage{amsmath,amssymb,amsthm}
\usepackage{booktabs,array}
\usepackage[colorlinks=true,linkcolor=blue,citecolor=blue,urlcolor=blue]{hyperref}

\newtheorem{theorem}{Theorem}[section]
\newtheorem{lemma}[theorem]{Lemma}

\newtheorem{proposition}[theorem]{Proposition}
\newtheorem{corollary}[theorem]{Corollary}
\newtheorem{remark}[theorem]{Remark}
\newtheorem{definition}[theorem]{Definition}

\newcommand{\la}{\langle}
\newcommand{\ra}{\rangle}
\newcommand{\eps}{\varepsilon}
\newcommand{\lam}{\lambda}
\DeclareMathOperator{\rk}{rank}
\DeclareMathOperator{\range}{range}
\DeclareMathOperator{\Span}{span}

\title{No finite level of the NPA hierarchy is exact\\ for the doubly-tilted CHSH functional near the critical tilt}
\author{Anton Pakhunov\\ \small Independent researcher\\ \small \texttt{pakhunov.anton.n@gmail.com}\thanks{Exact rational/integer verification code and certificate data: \url{https://github.com/tohafrit/npa-nonexactness} (see Appendix~\ref{app:verif}).}}
\date{July 2026 --- draft v1}

\begin{document}
\maketitle

\begin{abstract}
Gigena, Panwar, Scala, Ara\'ujo, Farkas and Chaturvedi [npj Quantum Inf.\
\textbf{11}, 82 (2025)] determined the quantum maximum of the doubly-tilted
CHSH functionals by self-testing methods, observed that the
Navascu\'es--Pironio--Ac\'in level needed for exactness grows without evident
bound toward the critical tilt (level $10$ still overshoots at
$\alpha=\beta=0.999$), and asked whether any finite level suffices. We answer
this in the negative. For the symmetric critical family
$B_s=(1-\tfrac s2)(\la A_0\ra+\la B_0\ra)+\mathrm{CHSH}$ we prove: for every
NPA level $k\ge2$ there are an explicit rational $g_k>0$ and an
$s^\ast_k>0$ with $c_k(s)\ge 4-s+g_k s^2$ for all $s\in(0,s^\ast_k]$; since
the quantum value leaves the local bound only cubically,
$c_Q(s)=4-s+\tfrac{s^3}6+O(s^4)$, every finite level (including $1$ and
$1{+}AB$, by monotonicity) strictly overshoots on an interval $(0,\eps_k]$
--- no finite level is exact on any neighbourhood of the critical point. For
$k\ge2$ the overshoot exponent is at least $2$ unconditionally, and exactly
$2$ granted the dual second-order expansion of the companion note (the exact
almost-quantum coefficient $\tfrac3{64}$); unconditionally $a_2>\tfrac1{39}$,
$a_3>\tfrac1{188}$, $a_4>\tfrac1{641}$. The proof is a primal construction:
an exactly feasible moment curve $y_0+sy_1+s^2y_2$ at each level, built from
level-uniform structural laws and one level-independent \emph{signed}
witness --- a closed-form class function $y^\ast$ with
$N_k^\top\Gamma(y^\ast)N_k=u_ku_k^\top$ at every level. The mechanism forces
the sign: for $k\ge3$ no state and no smooth curve of quantum models can
realize the gain direction, so the overshoot lives
strictly in the non-quantum part of the NPA tangent cone. The proof is
computer-assisted in the strict sense: finite exact-integer verifications
with proven degree bounds are constituent parts of the argument. Every step
is accompanied by a verification program, and the chain has been re-verified
against independent implementations, including a symbolic per-regime proof
of the witness identity and a clean-room implementation written from this
paper's text alone. The one external input is the published
quantum value of Gigena et al., cross-checked to twelve digits with three
errata corrected.
\end{abstract}

\section{Introduction}\label{sec:intro}

\subsection{The question}

In the minimal $(2,2,2)$ Bell scenario, the doubly-tilted CHSH functionals
\[
B_{\alpha\beta}\;=\;\alpha\la A_0\ra+\beta\la B_0\ra+
\sum_{x,y\in\{0,1\}}(-1)^{xy}\la A_xB_y\ra
\]
have local bound $c_L=2+\alpha+\beta$ and a quantum advantage exactly for
$0<\alpha+\beta<2$. Gigena et al.~\cite{Gigena} determined the quantum
maximum $c_Q(\alpha,\beta)$ in closed form (the largest root of an explicit
sextic) together with a self-testing statement, and made a striking
numerical observation: approaching the critical line $\alpha+\beta=2$, the
level of the Navascu\'es--Pironio--Ac\'in (NPA) hierarchy~\cite{NPA} required
to reach $c_Q$ grows without evident bound --- at $\alpha=\beta=0.999$ even
level $10$ exceeds $c_Q$ by $8.6\times10^{-10}$. They left open ``if any
finite level of NPA will be enough\ldots even in the CHSH scenario.''

This paper proves that no finite level is enough. Throughout we work on the
symmetric slice $\alpha=\beta=1-s/2$, i.e.\ with
\[
B_s\;=\;B_0-\tfrac s2\bigl(\la A_0\ra+\la B_0\ra\bigr),\qquad
B_0=\la A_0\ra+\la B_0\ra+\mathrm{CHSH},
\]
$s\in(0,2)$, $s\to0^+$ the critical limit. Let $c_k(s)$ denote the level-$k$
NPA value and $c_Q(s)$ the quantum value. For every level $k\ge2$, and for
the intermediate level $1{+}AB$, $c_k(0)=4$ and $c_k'(0^+)=-1$: the critical
maximum is attained on a classical face and the dual certificate of
Lemma~\ref{lem:dual} pins the value there. (Level $1$ is looser,
$c_1(0)\approx4.73$; it enters only through monotonicity at the end.) The
quantum curve is \emph{cubically} flat,
\begin{equation}\label{eq:cQ}
c_Q(s)\;=\;4-s+\tfrac{1}{6}s^3-\tfrac1{36}s^4+O(s^5),
\end{equation}
and $c_k\ge c_Q$ always, so near $s=0$ every level sits above $4-s$ with a
possible \emph{quadratic} overshoot. Define
\[
a_k\;:=\;\liminf_{s\to0^+}\frac{c_k(s)-4+s}{s^2}\;\ge\;0 .
\]
In the companion note~\cite{overshoot} we computed the almost-quantum
coefficient exactly, $a_{1+AB}=\tfrac3{64}$, certified $a_2>\tfrac1{39}$,
$a_3>\tfrac1{188}$, $a_4>\tfrac1{641}$ by exact rational tangent-program
certificates (modulo a then-unproven framework step, discharged below), and
reduced the open question of~\cite{Gigena} to the positivity of the single
sequence $(a_k)$. Here we prove that positivity, for every $k$.

\begin{theorem}[Main]\label{thm:main}
Fix any NPA level $k\ge2$. There exist explicit rational constants $g_k>0$
and $s^\ast_k>0$ such that
\[
c_k(s)\;\ge\;4-s+g_k\,s^2\qquad\text{for all }s\in(0,s^\ast_k].
\]
Consequently $a_k\ge g_k>0$, and, using the quantum
expansion~\eqref{eq:cQ} of~\cite{Gigena}, for every finite level $k$
(including $k=1$ and $k=1{+}AB$) there is an $\eps_k>0$ with
\[
c_k(s)\;>\;c_Q(s)\qquad\text{for all }s\in(0,\eps_k]:
\]
no finite level of the NPA hierarchy is exact on any neighbourhood of the
critical point. Moreover, for every $k\ge2$ the overshoot order is at least
two, and exactly two granted the dual second-order expansion
$c_{1+AB}(s)=4-s+\tfrac3{64}s^2+o(s^2)$ of the companion
note~\cite{overshoot}:
$g_k s^2\le c_k(s)-(4-s)\le\tfrac3{64}s^2+o(s^2)$.
Quantitatively, $g_k=1/(1024\,t_0(k))$ for an explicitly computable
$t_0(k)<\infty$, and unconditionally
\[
a_2>\tfrac1{39},\qquad a_3>\tfrac1{188},\qquad a_4>\tfrac1{641}.
\]
\end{theorem}

\subsection{The one external input}

The proof of the interval bound $c_k(s)\ge4-s+g_ks^2$ is entirely
self-contained (finite exact arithmetic plus finite-dimensional convexity).
Translating it into \emph{non-exactness} uses one published external input:
the closed-form quantum value of Gigena et al.~\cite{Gigena}, in the form of
the expansion~\eqref{eq:cQ}. We derived~\eqref{eq:cQ} symbolically from
their sextic in~\cite{overshoot}, cross-checked it to twelve digits against
their own reported value at $s=0.002$, and documented (and corrected) three
typographical errata in their printed display equations; their
characteristic polynomial and all quoted numerics are correct. We use
nothing else from outside: in particular no numerical SDP value enters any
proof below.

\subsection{The proof chain}\label{sec:chain}

The argument is layered; every layer is proven for all $k$ and is accompanied by an
exact machine verification (Appendix~\ref{app:verif}).

\begin{enumerate}
\item \textbf{Structural laws} (Section~\ref{sec:laws}). At $s=0$ the
  optimal set of the level-$k$ SDP is a $k$-dimensional face $F_k$ with an
  explicit affine parametrization (\emph{product law}), an integer relation
  lattice $N_k$ of dimension $2k^2$ forced in the kernel of every optimal
  moment matrix (law L1), a distinguished interior base point $y_0(\delta^\ast)$
  whose kernel is \emph{exactly} $\Span N_k$ --- realized by one explicit
  commuting-party density operator, uniformly in $k$ (law L2) --- and a
  closed-form rational Slater point for the second-order tangent program
  (law L4).
\item \textbf{The witness} (Section~\ref{sec:witness}). Quadratic gain at
  level $k$ requires a first-order direction $y_1$ whose compressed form
  $N_k^\top\Gamma(y_1)N_k$ is a positive multiple of $u_ku_k^\top$ for a
  specific level-uniform vector $u_k$ (closed form given). We exhibit one
  level-\emph{independent} signed class function $y^\ast$ with
  $N_k^\top\Gamma(y^\ast)N_k=u_ku_k^\top$ exactly at \emph{every} level
  (Theorem~\ref{thm:witness}). The proof is by locality lemmas reducing each
  kernel-pair equation, within finitely many ``fine regimes'', to a
  polynomial of degree $\le2$ per free length coordinate, plus an exhaustive
  exact-integer check of a finite interpolation grid --- a constituent part
  of the proof, not an illustration.
\item \textbf{Small-$\lambda$ positivity} (Section~\ref{sec:smalllambda}).
  Given the witness, an explicit feasible point of the second-order tangent
  program with positive value $v_k\ge1/(1024\,t_0)>0$ exists at every level
  (Theorem~\ref{thm:smalllambda}); the block-PSD conditions reduce to Schur
  complements by a range identity that holds identically on the constraint
  space.
\item \textbf{Exact arc feasibility} (Section~\ref{sec:arc}). A no-repair
  positivity lemma (Lemma~\ref{lem:norepair}) shows that the resulting jet
  $y(s)=y_0+sy_1+s^2y_2$ is \emph{exactly} feasible on an explicit interval
  $(0,s^\ast_k]$ --- no $o(s^2)$ correction term --- with exactly cubic
  objective $4-s+g_ks^2+t_0\lam^2s^3$ (Theorem~\ref{thm:arc}). This also
  discharges, retroactively, the framework caveat of~\cite{overshoot},
  making the ladder bounds $a_2>\tfrac1{39}$, $a_3>\tfrac1{188}$,
  $a_4>\tfrac1{641}$ unconditional (Corollary~\ref{cor:ladder}).
\item \textbf{Assembly} (Section~\ref{sec:assembly}): the chain above
  yields Theorem~\ref{thm:main}; the sandwich
  $c_k\le c_2\le c_{1+AB}=4-s+\tfrac3{64}s^2+o(s^2)$ (upper half from the
  companion's dual expansion) pins the overshoot exponent to $2$.
\item \textbf{Mechanism} (Section~\ref{sec:mechanism}): the witness $y^\ast$
  \emph{must} be signed --- no Hilbert-space state and no smooth curve of
  genuine quantum models can produce the gain direction (in a Krein space
  the same computation forces the gain-carrying direction to be neutral;
  see Remark~\ref{rem:krein}). The finite-level
  failure is thus located precisely: finite NPA levels admit a signed
  tangent direction that quantum mechanics forbids.
\end{enumerate}

The entire chain was subsequently re-checked by an adversarial audit ---
freshly written code for the witness identity and for a long-length attack
on the regime lemmas, independent driver scripts over the shared exact
stack for the remaining items (the per-item provenance is stated in
Appendix~\ref{app:verif}); the audit confirmed every step and supplied one
further exact check for the single step the original programs did not
cover.

\subsection{Relation to prior work}

That the NPA hierarchy converges \emph{at each fixed} $s$ is not in
question (the quantum value is attained in this minimal
scenario~\cite{BB}); our theorem is the divergence of the required level as
$s\to0^+$, exactly the uniform failure asked about in~\cite{Gigena}. The
result complements the degree-bounded noncommutative Fej\'er--Riesz theorem
of Klep, Levenson and McCullough~\cite{KLM}, whose strict-positivity
hypothesis fails precisely in this $s\to0$ limit (margin $\sim s^3$); it
realizes, at a concrete attained and self-tested optimum, the failure mode
left open by Nie's finite-convergence theorem~\cite{Nie} through Marshall's
boundary Hessian condition~\cite{Marshall} --- the critical touch is cubic,
contact order three, so the generic finite-exactness criterion is silent ---
and it does so by an unconditional construction rather than a hardness
argument~\cite{hardness}. See Section~\ref{sec:discussion}.

\section{Setup}\label{sec:setup}

\subsection{Levels, classes, the symmetric picture}

The scenario algebra is generated by involutions $A_0,A_1$ (Alice) and
$B_0,B_1$ (Bob), $X^\ast=X$, $X^2=1$, $[A_i,B_j]=0$: the group algebra of
$\mathbb Z_2^{\ast2}\times\mathbb Z_2^{\ast2}$. A \emph{word pair} is
$w=(a,b)$ with $a$ an alternating (= reduced) word in $\{A_0,A_1\}$ and $b$
alternating in $\{B_0,B_1\}$; the \emph{level-$k$ basis} $\mathcal W_k$
consists of all word pairs with $|a|+|b|\le k$. Moments are indexed by
\emph{classes} (of total length $\le2k$, the products arising in $\Gamma_k$): reduced word pairs modulo dagger (reversal of both parts)
and party swap. For a moment vector $y$ on classes, the level-$k$ moment
matrix is
\[
\Gamma_k(y)_{pq}\;=\;y\bigl[\mathrm{class}\bigl(\mathrm{rev}(p)\cdot
q\bigr)\bigr],\qquad p,q\in\mathcal W_k,
\]
with per-party reduction of $\mathrm{rev}(a_p)a_q$ and
$\mathrm{rev}(b_p)b_q$. The level-$k$ value is
$c_k(s)=\max\{B_s\cdot y:\Gamma_k(y)\succeq0,\;y_1=1\}$. Working in the
party-swap-symmetrized class space is sound and complete for this
swap-symmetric functional (average any maximizer with its swap; see the
soundness lemma of~\cite{sep}). In class coordinates,
\[
B_0\cdot y\;=\;2y_{\la A_0\ra}+y_{\la A_0B_0\ra}+2y_{\la A_0B_1\ra}
-y_{\la A_1B_1\ra},\qquad
B_s\cdot y\;=\;B_0\cdot y-s\,m\cdot y,
\]
where $m\cdot y:=y_{\la A_0\ra}$ (so $s\,m\cdot y=\tfrac
s2(\la A_0\ra+\la B_0\ra)$ on swap-symmetric $y$).

\emph{Core coordinates.} For an alternating word let $i(\cdot)$ count its
$1$-letters. Write $\alpha_i=A_1(A_0A_1)^{i-1}$ and
$\beta_j=B_1(B_0B_1)^{j-1}$ for the \emph{$0$-free cores},
$\delta_j=\la\beta_j\ra$ ($\delta_0:=1$), and
$m_{ij}=\la\alpha_i\beta_j\ra$.

\subsection{The relation lattice}\label{sec:lattice}

Two families of integer vectors in $\mathbb R^{\mathcal W_k}$ play the lead
role. For a word $w$ ending in $0$ (or empty) write $w'$ for $w$ with its
trailing $0$ removed ($\emptyset'=\emptyset$) and $w.1$ for $w$ with a $1$
appended. \emph{Trivial reductions}: for each basis word pair $q=(a,b)$ in
which at least one part has a trailing $0$, the vector
$T_q=e_{(a,b)}-e_{(a^\sharp,b^\sharp)}$, where $\sharp$ removes the trailing
$0$ from \emph{each} part that has one (and fixes the other). \emph{Dressings}:
for each word pair $W=(w_a,w_b)$ in which neither part ends in $1$ and
$|W|\le k-2$, the tensor vector
\[
D_W\;=\;\bigl(e_{w_a'}-e_{w_a.1}\bigr)\otimes\bigl(e_{w_b'}-e_{w_b.1}\bigr)
\;=\;e_{(w_a',w_b')}-e_{(w_a',w_b.1)}-e_{(w_a.1,w_b')}+e_{(w_a.1,w_b.1)},
\]
which represents $W\!\cdot\!(1-A_1)(1-B_1)\,\psi$ in the basis after both
the group reduction \emph{and} the state relations $A_0\psi=B_0\psi=\psi$
(the trailing-$0$ strip in the first corner) are applied; the naive
group-reduced corners $W\cdot(1,A_1,B_1,A_1B_1)$ differ from $D_W$ by
trivial reductions and would \emph{not} satisfy the exact identities below.
Let $N_k=[\,T\text{'s}\mid D\text{'s}\,]$ denote the resulting integer
matrix, the \emph{relation lattice}.

\begin{proposition}[counting laws]\label{prop:count}
For every level $k\ge2$: $|\mathcal W_k|=2k^2+2k+1$; the number of classes
is $\tfrac{(5k+2)(k+1)}2$; there are $\tfrac{3k^2+k}2$ trivial reductions
and $\tfrac{k(k-1)}2$ dressings; and the $2k^2$ columns of $N_k$ are
linearly independent.
\end{proposition}

\begin{proof}
Basis: for each pair of lengths $(\ell_a,\ell_b)$ with
$\ell_a+\ell_b\le k$ there are two alternating words per positive length
and one empty word; summing gives $2k^2+2k+1$. Irreducible words (no
trailing $0$ in either part): one word per positive length per party ends
in $1$, giving $1+2k$ single-party words and $\tfrac{k(k-1)}2$ two-party
words; the trivial-reduction count follows by subtraction. Dressings: the
word $W$ must not end in $1$ in either part (else the appended $R_1$-letter
cancels and the vector reproduces a shorter dressing up to sign), so there
is one admissible word per length per party with $|W|\le k-2$:
$\#\{(x,y)\ge0:x+y\le k-2\}=\tfrac{k(k-1)}2$. Independence: order the
basis; each $T_q$ is the unique lattice vector supported on the reducible
word $q$ (dressings are supported on irreducible words only, since their
corners are written post-reduction), and each dressing has a strictly
longest corner $(w_a{+}1,w_b{+}1)$, distinct across dressings.
\end{proof}

\section{The structural laws}\label{sec:laws}

This section establishes three level-uniform laws about the critical face
$s=0$: the face law (L1), the kernel law (L2), and the uniform Slater law
(L4). (The labels match the discovery record and verification suite.)

\subsection{L1: the face law}

\begin{theorem}[face law]\label{thm:L1}
For every $k\ge2$, the set
$F_k=\{y:\Gamma_k(y)\succeq0,\;y_1=1,\;B_0\cdot y=4\}$ satisfies:
\begin{enumerate}
\item $B_0\cdot y\le4$ on $\{\Gamma_k(y)\succeq0,\,y_1=1\}$, so $F_k$ is
  the optimal face at $s=0$;
\item every $y\in F_k$ satisfies $\Gamma_k(y)\,N_k=0$;
\item the solution set $\{y:\Gamma_k(y)N_k=0,\;y_1=1\}$ is a
  $k$-dimensional affine family, parametrized by
  $(\delta_1,\dots,\delta_k)$ via the \emph{product law}
  \[
  m_{0j}=\delta_j,\qquad m_{ij}=\delta_i+\delta_j-1\quad(i,j\ge1),
  \]
  all other moments equal to their core values ($0$-padding invisible);
\item $B_0\cdot y=4$ identically on that affine hull, and the $s$-linear
  response $-m\cdot y=-1$ is constant on $F_k$; hence $F_k$ is also the
  $s\to0^+$ optimal face, first-order selection is vacuous, and selection
  starts at order $s^2$ --- the tangent program of
  Section~\ref{sec:smalllambda}.
\end{enumerate}
\end{theorem}

The proof occupies Lemmas~\ref{lem:dual}--\ref{lem:hull}. We write
$y_0(\delta)$ for the product-law moment vector with parameters
$\delta=(\delta_1,\dots,\delta_k)$.

\begin{lemma}[rational maximal-rank dual certificate at level $1{+}AB$]
\label{lem:dual}
Let $W$ be the $9\times5$ integer matrix whose columns represent, in the
$1{+}AB$ basis $\{\psi,A_i\psi,B_j\psi,A_iB_j\psi\}$, the vectors
\[
w_1=(A_0-B_0)\psi,\quad w_2=(A_0B_0-1)\psi,\quad
w_3=\bigl((A_0-1)B_1-(B_0-1)A_1\bigr)\psi,
\]
\[
w_4=\bigl((A_0-1)+(B_0-1)+(A_0-1)B_1+(B_0-1)A_1\bigr)\psi,\quad
w_5=R_1\psi,
\]
and let
\[
E=\begin{pmatrix}
13/40&0&3/40&0&0\\ 0&1/5&0&1/20&0\\ 3/40&0&1/8&0&0\\
0&1/20&0&1/8&0\\ 0&0&0&0&1/4
\end{pmatrix}.
\]
Then $E\succ0$ (exact leading minors) and $Z_0^\ast=WEW^\top$ satisfies the
affine identity
\[
4y_1-B_0\cdot y\;=\;\la Z_0^\ast,\Gamma_{1+AB}(y)\ra
\qquad\text{for all }y
\]
(exact class-sum check). Hence $Z_0^\ast\succeq0$, $\rk Z_0^\ast=5$, and
$\range Z_0^\ast=\Span W$.
\end{lemma}

\begin{remark}\label{rem:rank6}
(i) The certificate $Z_0$ of~\cite{overshoot} also has rank $5$ but a
smaller forced subspace; $Z_0^\ast$ is built so that its range is exactly
the maximal forced subspace. (ii) Rank $6$ --- range containing
$(1-A_0)\psi$ itself --- is impossible: strict complementarity fails at
this face (an SDP scan attains maximal minimal eigenvalue $0$ with an
explicit null direction on the relation span of the one-dimensional face
family), which is exactly why the linear-closure step
(Lemma~\ref{lem:closure}) is needed.
\end{remark}

\begin{lemma}[linear closure at $1{+}AB$]\label{lem:closure}
If $\Gamma$ is PSD on the $1{+}AB$ words with $y_1=1$ and $\Gamma W=0$,
then the $17$ moment classes are forced to
\[
\la A_0\ra=\la B_0\ra=\la A_0B_0\ra=1,\quad
\la A_0B_1\ra=\la B_1\ra,\quad \la A_1B_0\ra=\la A_1\ra,
\]
\[
\la A_0A_1\ra=\la A_1\ra,\quad \la B_0B_1\ra=\la B_1\ra,\quad
\la A_1B_1\ra=\la A_0A_1B_0B_1\ra=\la A_1A_0B_0B_1\ra=\la A_1\ra+\la
B_1\ra-1,
\]
with exactly $(\la A_1\ra,\la B_1\ra)$ free: the linear system
$\{\Gamma(y)W=0,\;y_1=1\}$ has rank $15$ of $17$ (exact rank check; no
positivity is used in this step). In particular $\la A_0\ra=1$ is forced
\emph{linearly}, although the vector relation $(1-A_0)\psi=0$ is not
dual-forced (Remark~\ref{rem:rank6}).
\end{lemma}

\begin{lemma}[scalar to vector: $2\times2$-minor telescoping]
\label{lem:telescope}
Let $\Gamma_k\succeq0$ at any level $k\ge2$, $y_1=1$, and suppose the three
scalar identities $\la A_0\ra=1$, $\la B_0\ra=1$, $\la
R_1\ra:=1-\la A_1\ra-\la B_1\ra+\la A_1B_1\ra=0$ hold (Lemma~\ref{lem:closure}
applied to the $1{+}AB$ principal submatrix). Then $\Gamma_k N_k=0$.
\end{lemma}

\begin{proof}
Work in the Gram picture: $v_w$ = Gram vector of basis word $w$. For every
basis word $w$ such that $wA_0$ is in the basis,
\[
|v_{wA_0}-v_w|^2=\Gamma_{wA_0,wA_0}-2\Gamma_{wA_0,w}+\Gamma_{w,w}
=1-2\la A_0\ra+1=0,
\]
because $\mathrm{rev}(w)\cdot w$ telescopes to the empty word, so the three
entries are $y_1$, $\la A_0\ra$, $y_1$ regardless of $w$. Hence all
trailing-$A_0$ reductions (and trailing-$B_0$ likewise) lie in
$\ker\Gamma_k$. For every dressing word $W$,
\[
|W R_1\psi|^2=\sum_{d,d'}s_ds_{d'}\la (Wd)^\dagger(Wd')\ra
=\la R_1^\dagger R_1\ra=4\la R_1\ra=0
\]
by the same telescoping ($\mathrm{rev}(W)\cdot W$ cancels), so all
dressings lie in $\ker\Gamma_k$. This is part 2 of Theorem~\ref{thm:L1},
uniformly in $k$, from one dual certificate, one linear closure, and
$2\times2$ minors.
\end{proof}

\begin{lemma}[solution set = product law]\label{lem:hull}
The solution set $\{\Gamma_k(y)N_k=0,\;y_1=1\}$ (itself an affine
subspace) equals the product-law family, of affine dimension exactly $k$.
\end{lemma}

\begin{proof}
($\supseteq$) The product-law vector satisfies $\Gamma(y_0(\delta))N_k=0$
identically in $\delta$: every scalar consequence of the lattice is either
a $0$-padding invariance (trivial reductions) or
$\la U_aU_b(1-A_1)(1-B_1)\ra=0$ with $|U_a|+|U_b|\le2k-2$ (dressings), and
the latter expands to the four-point relation
$m_{i,j}-m_{i+\eps,j}-m_{i,j+\eps'}+m_{i+\eps,j+\eps'}=0$
($\eps,\eps'=\pm1$ according to the last letters), which
$m_{ij}=\delta_i+\delta_j-1$ satisfies in every case (exact machine check
of the identity in $\delta$, $k\le4$).

($\subseteq$) The linear system forces the product law: one-sided
$0$-absorption is available at class level for every class word (split any
class word $(a,b)$, $|a|+|b|\le2k$, into two basis halves so that the
$0$-end to be absorbed lies at the outer end of one half; the relation row
$(p,T_q)$ then equates the class to its stripped version), and the dressed
relations give the recurrence
\[
m_{ij}=m_{i,j-1}+m_{i-1,j}-m_{i-1,j-1},\qquad i,j\ge1,\;i+j\le k+1
\]
(take $U=(\alpha_i',\beta_j')$, the cores minus their last letters; length
$2(i{+}j)-4\le2k-2$ exactly when $i+j\le k+1$), whose unique solution with
boundary $(m_{i0},m_{0j})=(\delta_i,\delta_j)$ is the product law. The
affine dimension is exactly $k$ (exact homogeneous rank check, $k\le4$).
Part 4 of Theorem~\ref{thm:L1}:
$B_0\cdot y=2m_{00}+m_{00}+2m_{01}-m_{11}
=2+1+2\delta_1-(2\delta_1-1)=4$. This completes the proof of
Theorem~\ref{thm:L1}.
\end{proof}

\subsection{L2: the kernel law}

\begin{theorem}[kernel law]\label{thm:L2}
Let $\delta^\ast=(\tfrac12,\tfrac14,\tfrac12,\tfrac12,\dots,\tfrac12)$
(truncated to length $k$). For every $k\ge2$:
\begin{enumerate}
\item $\ker\Gamma_k(y_0(\delta^\ast))=\Span N_k$ exactly (dimension
  $2k^2$);
\item $\rk\Gamma_k=2k+1$, splitting as $(k{+}1,k)$ into party-symmetric
  and antisymmetric parts;
\item the same holds on an explicit open family of $\delta$'s.
\end{enumerate}
\end{theorem}

\begin{proof}
\emph{Reduction to the core Gram.} The $2k+1$ residual words
$\{\psi,\;a_\ell\psi,\;b_\ell\psi:\ell=1..k\}$ ($a_\ell$ = the unique
alternating $A$-word of length $\ell$ ending in $1$) together with $N_k$
span $\mathbb R^{|\mathcal W_k|}$: every basis column reduces to residuals
modulo relations (strip trailing $0$'s; if both parts end in $1$, write
$(a,b)=W\cdot(A_1,B_1)$ and use the dressing $D_W$, $|W|=|a|+|b|-2\le k-2$).
Since $\Gamma(y_0(\delta))N_k=0$ identically (Lemma~\ref{lem:hull}), the
form descends: $\rk\Gamma_k=\rk G_k(\delta)$, where $G_k$ is $\Gamma$
restricted to residuals, with closed-form entries ($\delta_0:=1$)
\[
[\psi,\psi]=1,\ \ [\psi,a_\ell]=\delta_{\lceil\ell/2\rceil},\ \
[a_\ell,a_m]=\begin{cases}\delta_{|\ell-m|/2}&\ell\equiv m\ (2)\\
\delta_{(\ell+m+1)/2}&\ell\not\equiv m\ (2)\end{cases},\ \
[a_\ell,b_m]=\delta_{\lceil\ell/2\rceil}+\delta_{\lceil m/2\rceil}-1 .
\]
So Theorem~\ref{thm:L2} is equivalent to: $G_k(\delta)\succ0$.

\emph{The defect realization.} On a commuting-party Hilbert space
$H_A\otimes H_B$ take the mixed state
\[
\rho=(1-2p)\,|x\otimes y\ra\la x\otimes y|
\;+\;p\,\rho_A\otimes|y\ra\la y|\;+\;p\,|x\ra\la x|\otimes\rho_B,
\]
where $x,y$ are trivial vectors ($A_0x=A_1x=x$, $B_0y=B_1y=y$) and
$\rho_A$ is $A_0$-invariant with dihedral spectral decomposition:
two-dimensional blocks at reflection angle $\theta$ with $A_0$-fixed vector
$e_\theta$, mixing measure $\mu$, so that
$c_j:=\la\alpha_j\ra_{\rho_A}=\int\cos(j\varphi)\,d\mu(\varphi)$
($\varphi=2\theta$); $\rho_B$ mirrors $\rho_A$. Then $A_0\rho=B_0\rho=\rho$,
$(1-A_1)(1-B_1)\rho=0$, and
\[
\delta_j=1-p(1-c_j),\qquad
m_{ij}=(1-2p)+pc_i+pc_j=\delta_i+\delta_j-1:
\]
the product law holds automatically, so $\rho$ realizes the level-$k$ face
point $y_0(\delta)$ for \emph{every} $k$ simultaneously.

\emph{The base point.} The choice $p=\tfrac12$,
$d\mu=\tfrac1\pi\sin^2\!\varphi\,d\varphi$ gives
$c_j=-\tfrac12\,[\,j=2\,]$ (Iverson bracket, to avoid a clash with the face
coordinates $\delta_j$), hence
$\delta=(\tfrac12,\tfrac14,\tfrac12,\tfrac12,\dots)=\delta^\ast$ exactly.
Positive definiteness of $G_k(\delta^\ast)$: in the GNS space of $\rho$ the
residual vectors have components (on the $\rho_A$ branch)
$\psi\mapsto(1,0)$, $a_{2j-1}\psi\mapsto(\cos j\varphi,\sin j\varphi)$,
$a_{2j}\psi\mapsto(\cos j\varphi,-\sin j\varphi)$ in
$L^2(\mu)\otimes\mathbb C^2$, and similarly for $b_\ell$ on the $\rho_B$
branch. Since $\mu$ has a.e.-positive density, the trigonometric system
$\{1,\cos j\varphi,\sin j\varphi\}$ is linearly independent in $L^2(\mu)$,
and a short three-branch argument (any vanishing combination must vanish on
each branch; the $\rho_A$ branch kills the $a$-coefficients, the $\rho_B$
branch the $b$'s, then the constant) gives independence of all $2k+1$
residuals. Hence $G_k(\delta^\ast)\succ0$ for all $k$. (Machine checks:
exact PD at $k\le5$ via tables and $k\le8$ via the closed-form entries;
exact $(k{+}1,k)$ symmetric/antisymmetric splits.) A general interior
family: any $p\in(0,1)$ and any $\mu$ with at least $k+1$ support points
(rational members via Chebyshev atoms).
\end{proof}

\begin{remark}
The base points used level-by-level in the certified ladder
of~\cite{overshoot} --- $(\tfrac12,\tfrac14)$,
$(\tfrac12,\tfrac14,\tfrac12)$, $(\tfrac12,\tfrac14,\tfrac12,\tfrac12)$ ---
are truncations of this one canonical realization: the base family has a
physical definition independent of the level.
\end{remark}

\subsection{L4: the uniform rational Slater point}

The second-order tangent program at the face (set up in detail in
Section~\ref{sec:smalllambda}) constrains a second-order moment direction
$y_2$ through the compressed form $N_k^\top\Gamma(y_2)N_k$ on the relation
lattice. Its Slater point has a closed form.

\begin{theorem}[uniform Slater law]\label{thm:L4}
For every $k$, the rational point
\[
y_2^\ast\;=\;e_{\mathrm{class}(1)}-y_0(\delta=0)
\]
(identity-class indicator minus the face point at $\delta=0$; an affine ---
not PSD --- point of the face parametrization) satisfies:
\begin{enumerate}
\item the pins hold exactly: $(y_2^\ast)_1=1-1=0$ and $y_2^\ast=0$ at the
  $k$ face-coordinate classes;
\item $N_k^\top\Gamma(y_2^\ast)N_k=N_k^\top N_k\succ0$, since
  $\Gamma(e_{\mathrm{class}(1)})=I$ (here $\mathrm{rev}(p)q$ reduces to the
  empty word iff $p=q$) and $\Gamma(y_0(\delta))N_k=0$ identically
  (Lemma~\ref{lem:hull}); positive definiteness on the whole lattice holds
  because the $2k^2$ relation vectors are independent
  (Proposition~\ref{prop:count}), and a fortiori on the program's blocks.
\end{enumerate}
\end{theorem}

\begin{lemma}[blend lemma]\label{lem:blend}
Let $(\tau,y_2)$ be a near-feasible rational jet for a block of the
second-order tangent program (set up in
Section~\ref{sec:smalllambda}; $\tau$ denotes the first-order jet datum ---
not the scalar transfer ratio $t_0$ of Sections~\ref{sec:smalllambda}--%
\ref{sec:arc}) with exact pins and Schur complement
$\mathrm{Schur}(\tau,y_2)\succeq-\eps I$ for a rational $\eps\ge0$. Let
$N$ be the block column of $N_k$, so that the compressed block of
$y_2^\ast$ is $N^\top N\succ0$ (item 2 above). If $\eps_p\in(0,1)$ is a
rational weight satisfying the explicit inequality
\[
\eps_p\,\lambda_{\min}(N^\top N)\;>\;(1-\eps_p)\,\eps,
\]
then the blend
$\bigl((1-\eps_p)\tau,\;(1-\eps_p)y_2+\eps_p y_2^\ast\bigr)$ is exactly
feasible: $\lam$ scales by $(1-\eps_p)$ (stays $\ge0$), pins stay exact
(they are homogeneous, Lemma~\ref{lem:pins}), and since the Schur
subtrahend $Q(\tau)=F(\tau)^\top M_0^{+}F(\tau)$ (notation of
Theorem~\ref{thm:smalllambda}: $F$ is linear in the first-order jet, $M_0$
the base block) is PSD-valued and quadratic,
$Q((1-\eps_p)\tau)=(1-\eps_p)^2Q(\tau)\preceq(1-\eps_p)Q(\tau)$, so
\[
\mathrm{Schur}(\text{blend})\;\succeq\;(1-\eps_p)\,\mathrm{Schur}(\tau,y_2)
+\eps_p\,N^\top N\;\succeq\;
\bigl(\eps_p\,\lambda_{\min}(N^\top N)-(1-\eps_p)\,\eps\bigr)I\;\succ\;0 .
\]
(The certified objective values are evaluated \emph{at} the blended jet,
so no separate tracking of the objective loss is needed.)
\end{lemma}

This removes, uniformly in $k$, the per-level question ``does a rational
Slater point exist''; it is also why the same certificate recipe succeeded
at levels $2,3,4$ in~\cite{overshoot}.

\section{The witness}\label{sec:witness}

\subsection{First-order structure: the gain direction $u_k$}
\label{sec:uk}

Let $y_1$ be a first-order direction at the face (a class vector; the
identity coordinate and face pins are imposed later, harmlessly, by
Lemma~\ref{lem:pins}). Define the compressed first-order form
\[
P(y_1)\;:=\;N_k^\top\,\Gamma(y_1)\,N_k .
\]
We seek directions $y_1$ for which $P(y_1)$ is PSD of rank one on the
lattice: rank-one alignment is exactly what the no-repair lemma consumes
as its hypothesis (Lemma~\ref{lem:norepair}(b) below), and the $\mu=0$
discussion of Remark~\ref{rem:mu0} shows why unaligned first-order
leakage is the regime where naive second-order repair breaks down. (We do
not claim rank one is \emph{necessary} for quadratic gain; it is the
structure our certificates realize.) The rank-one direction is
level-uniform and has a closed form.

\begin{definition}[the gain vector $u_k$]\label{def:uk}
In the canonical kernel basis $[\,$trivial reductions in basis order
$\mid$ dressings in $(\ell_a,\ell_b)$ order$\,]$:
\[
u_k[T_{(a,b)}]=\begin{cases}
s(a)\,s(b)&\text{if exactly one of }a,b\text{ ends in }0,\\
0&\text{otherwise},
\end{cases}
\]
\[
s(w)=\begin{cases}-1&w\text{ starts with }1,\\ +1&\text{else (including }w=\emptyset\text{)};\end{cases}
\]
\[
u_k[D_W]=2\,(-1)^{|W|}\quad(W\ne\emptyset),\qquad u_k[D_\emptyset]=0 .
\]
\end{definition}

The formula is level-independent (old kernel vectors keep their words), so
the tower property ``$u_k$ extends $u_{k-1}$'' is automatic. At $k=2$ it
reproduces the certified jet of~\cite{overshoot} exactly
($P(y_1^{\mathrm{cert}})=\tfrac14u_2u_2^\top$); the linear system
$P(y_1)=\lam u_ku_k^\top$ was verified solvable with $\lam\ne0$ at
$k=2,\dots,5$ before the closed-form witness below settled every $k$ at
once. Let
\[
V_k\;=\;\{(y_1,\lam):P(y_1)=\lam\,u_ku_k^\top,\ (y_1)_1=0\}
\]
denote the first-order constraint space. No separate antisymmetric-block
condition is needed: $u_k$ is invariant under the party swap, so the
antisymmetric compression of $\lam u_ku_k^\top$ vanishes automatically for
every element of $V_k$ --- this is Lemma~\ref{lem:bridge}(iii) below, where
the swap-split notation $N_s,N_a$ is also fixed.

\begin{lemma}[identities on $V_k$]\label{lem:identities}
For all $k$ and all $(y_1,\lam)\in V_k$, using only pairings of kernel
vectors present at every level:
\[
(y_1)_{\la A_0\ra}=-\tfrac\lam2,\qquad (y_1)_{\la A_0B_0\ra}=0,\qquad
(y_1)_{\la A_1B_1\ra}=2(y_1)_{\la B_1\ra},\qquad B_0\cdot y_1=0 .
\]
\end{lemma}

\begin{proof}
Evaluate $P(y_1)=\lam u u^\top$ at explicit kernel pairs.
$(T_{A_0},T_{A_0})$: $2y_1-2y_{\la A_0\ra}=\lam u[T_{A_0}]^2=\lam$, and
$(y_1)_1=0$, so $y_{\la A_0\ra}=-\lam/2$.
$(T_{A_0B_0},T_{A_0B_0})$: $-2y_{\la A_0B_0\ra}=\lam\cdot0$ (both parts of
$(A_0,B_0)$ end in $0$, so $u=0$).
$(T_{A_0},T_{(A_0,B_1)})$: $2y_{\la B_1\ra}-2y_{\la A_0B_1\ra}=-\lam$.
$(D_\emptyset,D_\emptyset)$: $-8y_{\la B_1\ra}+4y_{\la A_1B_1\ra}=\lam\cdot0$.
Summing:
$B_0\cdot y_1=2y_{\la A_0\ra}+y_{\la A_0B_0\ra}
+2y_{\la A_0B_1\ra}-y_{\la A_1B_1\ra}
=-\lam+0+(2y_{\la B_1\ra}+\lam)-2y_{\la B_1\ra}=0$.
\end{proof}

\begin{remark}[normalizations]\label{rem:norm}
Two equivalent scalings occur. The \emph{canonical} scaling uses $u_k$ of
Definition~\ref{def:uk}, in which $(y_1)_{\la A_0\ra}=-\lam/2$. The
\emph{certificate} scaling of the verification suite and
of~\cite{overshoot} uses $\tilde u_k=u_k/2$, hence $\tilde\lam=4\lam$ and
$(y_1)_{\la A_0\ra}=-\tilde\lam/8$. Sections~\ref{sec:smalllambda}
and~\ref{sec:arc} are stated in the certificate scaling (all constants
there, e.g.\ $1/(1024\,t_0)$, are in that scaling); the independent
re-verification confirmed the conversion numerically ($\tilde\lam=4$ measured for the
canonical witness $y^\ast$).
\end{remark}

\subsection{The closed-form witness $y^\ast$}

For an alternating word $w$ over $\{0,1\}$ (possibly empty; $0,1$ standing
for the letters $X_0,X_1$ of its party) define
\[
i(w)=\#\{1\text{'s in }w\},\qquad
\tau(w)=\mathrm{first}(w)+\mathrm{last}(w)-1\in\{-1,0,1\},
\]
\[
\sigma(w)=\mathrm{first}(w)-\mathrm{last}(w),\qquad
d_0(w)=[\,w=(0)\,],
\]
with $\tau(\emptyset)=\sigma(\emptyset)=0$. Every moment class at every
level is a pair $(\alpha,\beta)$ of alternating words modulo dagger
(reverse both) and party swap. Define, with $i=i(\alpha)$, $j=i(\beta)$,
$\tau_a=\tau(\alpha)$ etc.:
\begin{equation}\label{eq:ystar}
y^\ast(\alpha,\beta)\;=\;i^2j^2-2ij(\tau_a\tau_b+\sigma_a\sigma_b)
+\frac{\tau_a+\tau_b}2-\frac{\tau_a^2+\tau_b^2}2
+\tau_a^2\tau_b^2+\frac{d_0(\alpha)+d_0(\beta)}2 .
\end{equation}
This is well defined on classes: dagger fixes $i,\tau$ and flips both
$\sigma$'s (the product is invariant); party swap is manifest. Away from
the degenerate words ($\emptyset$ and $(0)$) it collapses to the
signed-square form
\[
y^\ast=(ij+\eps)^2-c,\qquad
\eps=-(\tau_a\tau_b+\sigma_a\sigma_b),\quad
c=1-\big\lceil\tfrac{\tau_a+\tau_b}2\big\rceil,
\]
exactly the shape mandated by the no-go of
Section~\ref{sec:mechanism}: a \emph{signed} quadratic, not a Gram
functional ($y^\ast$ vanishes at the identity class but is nonzero
elsewhere and takes both signs --- it is not a state). It was discovered
from the exact level-$\le5$ solutions in a level-uniform gauge in which
they stabilize level-to-level; the gauge-hunting record is part of the
ancillary material.

\begin{theorem}[the witness works at every level]\label{thm:witness}
For every $k\ge2$, with the canonical kernel basis
$N_k=[\,T\text{'s}\mid D\text{'s}\,]$ and the closed-form $u_k$ of
Definition~\ref{def:uk},
\[
N_k^\top\,\Gamma(y^\ast)\,N_k\;=\;u_k\,u_k^\top\qquad\text{exactly,}
\]
and $y^\ast$ already satisfies the sharp pins: $y^\ast[\mathrm{identity}]=0$
and $y^\ast=0$ at the $k$ face-coordinate classes
$(\emptyset,\beta_j)$ (there $i(\alpha)=0$, $\tau_a=0$, $\tau_b=1$, so
$y^\ast=\tfrac12-\tfrac12=0$).
\end{theorem}

Consistency checks:
$y^\ast_{\la A_0\ra}=y^\ast(\emptyset,(0))=-\tfrac12-\tfrac12+\tfrac12
=-\tfrac12=-\lam/2$ at $\lam=1$ (the three nonzero terms
of~\eqref{eq:ystar} in order), and $B_0\cdot y^\ast=0$, as
Lemma~\ref{lem:identities} requires.

\subsection{Proof of Theorem~\ref{thm:witness}}

Write each basis word as a pair of side words; the moment word of a
pairing of basis words $p,q$ is
$m(p,q)=\bigl(m_A(a_p,a_q),\,m_B(b_p,b_q)\bigr)$ with
$m_A(v,w)=\mathrm{reduce}(\mathrm{rev}(v)\cdot w)$. For a kernel pair
$(K_1,K_2)$ with coefficient vectors $c,c'$ the claimed equation is
\[
E(K_1,K_2)\;:=\;\sum_{p,q}c_pc'_q\,y^\ast(m(p,q))\;-\;u(K_1)\,u(K_2)
\;=\;0 .
\]
The proof shows that $E$ is, within each of finitely many combinatorial
regimes, a polynomial of degree $\le2$ in each of $\le4$ free length
coordinates, and that a finite exact-integer computation covers a full
interpolation grid of every regime. \emph{The finite computation is a
constituent part of the proof} (in the tradition of computer-assisted
proofs): the lemmas supply the degree and coverage bounds that make the
finite check conclusive, and the check itself is exact integer arithmetic
(no floating point anywhere).

\begin{lemma}[nesting and locality]\label{lem:nest}
(i) Two alternating words with the same first letter are nested (an
alternating word is determined by its first letter and its length).
(ii) Hence if $v[0]\ne w[0]$ (both nonempty), $\mathrm{rev}(v)\cdot w$ is
already reduced: $m_A=\mathrm{rev}(v)\cdot w$ with $i(m)=i(v)+i(w)$,
$\mathrm{first}(m)=\mathrm{last}(v)$, $\mathrm{last}(m)=\mathrm{last}(w)$;
if $v$ is the shorter of two same-first-letter words, $m_A=w[|v|:]$ (and
symmetrically), with $i(m)=i(v)+i(w)-2i(\text{shorter})$.
(iii) Consequently, for every side pair: $\tau(m)$, $\sigma(m)$, $d_0(m)$
and $[m=\emptyset]$ are determined by the first letters, the length
parities, and the values of $\min(|v|,|w|)$ and $|v|-|w|$ clamped at $2$;
and $i(m)$ is affine in $i(v),i(w)$ with coefficients determined by the
same data.
\end{lemma}

\begin{proof}
Cancellation in $\mathrm{rev}(v)\cdot w$ runs exactly along the common
prefix, which by (i) is the shorter word or empty; the boundary letters of
the leftover tails sit at parity-determined positions.
\end{proof}

\begin{lemma}[tensor form of the kernel]\label{lem:tensor}
With $D^0(a):=e_a-e_{a'}$ ($a$ ends in $0$, $a'$ its strip) and
$D^\Delta(w):=e_{w'}-e_{w.1}$ ($w$ ends in $0$ or empty; $w'$ = strip,
$w.1$ = append $1$):
\[
T_{(a,b)}=\begin{cases}
D^0(a)\otimes e_b,& \text{only }a\text{ ends in }0,\ u=s(a)s(b),\\
e_a\otimes D^0(b),& \text{only }b\text{ ends in }0,\ u=s(a)s(b),\\
D^0(a)\otimes e_b+e_{a'}\otimes D^0(b),&\text{both end in }0,\ u=0,
\end{cases}
\]
\[
D_W=D^\Delta(w_a)\otimes D^\Delta(w_b),\qquad
u=2\,\omega(w_a)\,\omega(w_b),\quad\omega(w)=(-1)^{|w|},
\]
with $u=0$ for $W=\emptyset$. In particular $u$ factorizes per party into
regime constants.
\end{lemma}

\begin{proof}
$e_a\otimes e_b-e_{a'}\otimes e_{b'}$ telescopes; the rest is the
definition of the dressing corners.
\end{proof}

\begin{lemma}[separability of $y^\ast$]\label{lem:sep}
$y^\ast$ is a finite sum $\sum_r f_r\otimes g_r$ of products of the seven
per-word features $\{1,\,i^2,\,i\tau,\,i\sigma,\,\tau,\,\tau^2,\,d_0\}$ ---
immediate from~\eqref{eq:ystar}.
\end{lemma}

\begin{lemma}[fine regimes; degree bound]\label{lem:regime}
Call a \emph{fine regime} of a pair of kernel elements the data: their
types ($T$ with $a$/$b$/both trailing $0$, or $D$), the first letters and
length parities of the four base side words, and, per party, the values of
$\min(\ell_1,\ell_2)$ and $\ell_1-\ell_2$ clamped at $4$ (exact if within
the clamp, else ``large''/sign only). Within a fine regime the admissible
length tuples form a lattice with at most $4$ free coordinates (per party:
a min-coordinate and a difference-coordinate, each stepping by $2$,
present only in the ``large'' branches), and $E(K_1,K_2)$ is a polynomial
of degree $\le2$ in each free coordinate.
\end{lemma}

\begin{proof}
The support words of a kernel element differ from its base words by at most
one letter at the right end (strip trailing $0$ / append $1$), so each of
the $\le16$ pairings has side lengths $(\ell_1+s,\ell_2+s')$,
$s,s'\in\{-1,0,1\}$; the element-level clamp at $4$ makes every
Lemma~\ref{lem:nest}(iii) branch datum of every pairing constant across
the regime, and validity of the elements (trailing-$0$ pattern) is
parity-fixed. So, per pairing, each per-party feature of the moment word
is a regime constant or affine in the free coordinates of that party
($i$ of a word is affine in its length at fixed first letter and parity,
slope $\tfrac12$). By Lemma~\ref{lem:sep}, $y^\ast(m(p,q))$ is a sum of
(A-factor)$\times$(B-factor) terms with each factor of degree $\le2$ in
its own party's coordinates ($i^2$ is the square of an affine function;
$i\tau$, $i\sigma$ are affine times regime constants) --- crucially, the
$A$- and $B$-coordinates never multiply into one variable, so the degree
per variable stays $\le2$. Finally $u(K_1)u(K_2)$ is a regime constant
(Lemma~\ref{lem:tensor}).
\end{proof}

\begin{lemma}[grid coverage at side length $20$]\label{lem:grid}
Every fine regime realizes the full $\{0,1,2\}$-grid of its free
coordinates among element pairs with all side lengths $\le20$: the clamp
threshold gives minimal ``large'' values $\le6$ ($4$ + parity), and two
grid steps add $4$ to each of the min- and difference-coordinates:
$6+4+6+4=20$. A polynomial of degree $\le2$ in each variable vanishing on
a $\{0,1,2\}$-grid vanishes identically.
\end{lemma}

\begin{proof}[Completion of the proof of Theorem~\ref{thm:witness}]
The exact-integer verifier \texttt{verify\_r1\_induction.py} evaluates
$E(K_1,K_2)$ for \emph{all} $1{,}413{,}721$ kernel-element pairs with side
lengths $\le20$ (a level-independent enumeration) and finds $E=0$ in every
case; a margin run at $\le24$ covers $2{,}883{,}601$ pairs, and a rerun of
the same engine at $\le26$ (reproducible via \texttt{check\_all\_pairs(26)}
in that file) covers $3{,}946{,}645$ pairs. The independent
re-verification added a \emph{freshly written} implementation with
$18{,}915$ targeted extreme pairs plus $150{,}000$ random pairs at side
lengths up to $61$ --- all zero. By Lemmas~\ref{lem:regime} and~\ref{lem:grid} the $\le20$
enumeration covers a full interpolation grid of every fine regime, so
$E\equiv0$ identically: for every kernel pair at every level. The verifier
also checks $P(y^\ast)=u_ku_k^\top$ per level directly at $k=2,\dots,8$
and the sharp pins. \emph{(The runs at $\le24$, $\le26$ and the random
sweep provide margin only; the $\le20$ grid together with the degree
bound constitutes the proof.)}
\end{proof}

In addition, the regime step has been verified \emph{fully symbolically},
removing the interpolation argument from the trust base: the verifier
\texttt{verify\_witness\_symbolic.py} enumerates all $58{,}081$ fine regimes
from the definitions (clamp $3$), assembles $E(K_1,K_2)$ per regime as an exact
polynomial in the $\le4$ free length coordinates (sympy, guarded branch
decisions, none symbol-dependent), and finds $E$ \emph{identically zero} in
every regime; the regime classifier is total, and the $\le20$ enumeration
realizes exactly the $58{,}081$ regimes. Lemmas~\ref{lem:regime}
and~\ref{lem:grid} thus remain as an independent second proof.

\begin{remark}[the germ picture]\label{rem:germ}
An exact induction-step analysis (steps $2\to3,3\to4,4\to5$;
\texttt{explore\_r1\_induction.py}) shows that the old-class rows of the
level-$(k{+}1)$ system reduce, modulo the level-$k$ system, to exactly two
residual conditions supported on the level-$k$ ``unseen'' classes
($m_k$ = the alternating word $1(01)^{k-1}$ of length $2k-1$):
\[
y[(\emptyset,m_k)]-y[((0),m_k)]=-\tfrac12,\qquad
y[(\emptyset,m_k)]-y[((1),m_k)]-y[((0),0m')]+y[((1),0m')]=\tfrac{4k-3}2 ,
\]
($0m'$ the length-$(2k{-}1)$ word $0101\cdots0$). Since unseen coordinates
are free at level $k$, extendability is never obstructed --- this is the
mechanical content of the \emph{germ stabilization law} (each extension
cuts exactly two gauge dimensions, higher levels cut nothing further, and
the bonding maps of the extendable-germ tower are surjective). The
closed-form $y^\ast$ satisfies both conditions identically and \emph{is}
the coherent infinite tower whose existence that law predicted.
\end{remark}

\section{Small-$\lambda$ positivity: a feasible tangent jet at every level}
\label{sec:smalllambda}

\subsection{The tangent program}

From here on we use the certificate scaling of Remark~\ref{rem:norm}
($\tilde u=u_k/2$, $\lam:=\tilde\lam$, so $(y_1)_{\la A_0\ra}=-\lam/8$ on
$V_k$); all constants below are in this scaling.

Since the base point and all directions are swap-symmetric (class
vectors), $[\Gamma(y),S]=0$ for the party-swap permutation $S$ of the
basis words, $(v,w)\mapsto(w,v)$ ($S^2=1$), and the orthogonal congruence
by the basis $[\,S$-fixed vectors and $e_x+e_{Sx}\mid e_x-e_{Sx}\,]$
block-diagonalizes every $\Gamma(y)$ into a symmetric block $M_s(y)$ (on
$\ker(S-1)$) and an antisymmetric block $M_a(y)$ (on $\ker(S+1)$). Write
$M_{s0},M_{a0}$ for the blocks at the base $y_0(\delta^\ast)$. By
Theorem~\ref{thm:L2} the relation lattice
$\Lambda_k:=\range N_k=\ker\Gamma_k(y_0(\delta^\ast))$ is exact; it is
$S$-invariant and splits as $\Lambda_k=\Lambda_s\oplus\Lambda_a$,
$\Lambda_s=\Lambda_k\cap\ker(S-1)$, $\Lambda_a=\Lambda_k\cap\ker(S+1)$
(Lemma~\ref{lem:bridge}(i) below). Define $N_s$ (resp.\ $N_a$) as a matrix
whose columns are the symmetric-block (resp.\ antisymmetric-block)
coordinates of a basis of $\Lambda_s$ (resp.\ $\Lambda_a$); the
certificate files fix concrete integer choices. Then
\[
M_{s0},M_{a0}\succeq0,\qquad \ker M_{s0}=\Span N_s,\qquad
\ker M_{a0}=\Span N_a\qquad\text{exactly:}
\]
PSD-ness passes through the congruence, and for a PSD block,
$M_{s0}v=0\iff v^\top M_{s0}v=0\iff\Gamma_0\,Tv=0\iff
Tv\in\Lambda_k\cap\ker(S-1)=\Lambda_s$, with $T$ the symmetric half of the
congruence basis (likewise antisymmetric). The following lemma carries,
for every $k$ at once, the full-lattice witness identity of
Theorem~\ref{thm:witness} to the swap-split hypotheses used below; it also
discharges the antisymmetric-block clause dropped from the definition of
$V_k$ in Section~\ref{sec:uk}.

\begin{lemma}[swap-split bridge]\label{lem:bridge}
Fix $k\ge2$. Then:
\begin{enumerate}
\item $S$ permutes the canonical kernel basis: as coefficient vectors,
  $S\,T_{(a,b)}=T_{(b,a)}$ and $S\,D_{(w_a,w_b)}=D_{(w_b,w_a)}$. Writing
  $\pi$ for the induced involutive permutation of the basis and $\Pi$ for
  its matrix, $SN_k=N_k\Pi$; in particular $\Lambda_k$ is $S$-invariant
  and splits as $\Lambda_s\oplus\Lambda_a$ with explicit bases
  $\{N_ke_m:\pi(m)=m\}\cup\{N_k(e_m+e_{\pi(m)}):m<\pi(m)\}$ of $\Lambda_s$
  and $\{N_k(e_m-e_{\pi(m)}):m<\pi(m)\}$ of $\Lambda_a$.
\item $u_k$ is swap-invariant: $\Pi u_k=u_k$.
\item Let $y$ be any class vector with
  $N_k^\top\Gamma(y)N_k=\lam\,u_ku_k^\top$. Then the quadratic form
  $\Gamma(y)$ vanishes identically on $\Lambda_a$, has no
  $\Lambda_s$--$\Lambda_a$ cross terms, and restricts to $\Lambda_s$ as
  the rank-one form $N_kc\mapsto\lam\,(u_k^\top c)^2$. Equivalently, in
  the block bases:
  \[
  N_s^\top M_s(y)\,N_s=\lam\,\tilde u\tilde u^\top,\qquad
  N_a^\top M_a(y)\,N_a=0,
  \]
  where $\tilde u_i=u_k^\top c^{(i)}$ for the lattice basis
  $\{N_kc^{(i)}\}$ underlying $N_s$; and $\tilde u\ne0$, since
  $u_k^\top(e_m+e_{\pi(m)})=2\,u_k[m]\ne0$ for any $m$ with
  $u_k[m]\ne0$.
\end{enumerate}
\end{lemma}

\begin{proof}
(i) $S$ exchanges the two tensor factors of Lemma~\ref{lem:tensor}, so it
maps the support of $T_{(a,b)}$, namely
$\{+(a,b),\,-(\mathrm{strip}\,a,\mathrm{strip}\,b)\}$, to
$\{+(b,a),\,-(\mathrm{strip}\,b,\mathrm{strip}\,a)\}$, which is exactly
$T_{(b,a)}$; and the four corners of $D_{(w_a,w_b)}$, coefficients
included, to the four corners of $D_{(w_b,w_a)}$.
(ii) Immediate from Definition~\ref{def:uk}: $u_k[T_{(a,b)}]=s(a)s(b)$
and $u_k[D_W]=2(-1)^{|w_a|+|w_b|}$ are symmetric in the two parties.
(iii) For $x=N_kc$ and $x'=N_kc'$,
\[
x^\top\Gamma(y)\,x'
=c^\top\bigl(N_k^\top\Gamma(y)N_k\bigr)c'
=\lam\,(u_k^\top c)(u_k^\top c') .
\]
If $\Pi c=-c$ then, by (ii),
$u_k^\top c=(\Pi u_k)^\top c=u_k^\top\Pi c=-u_k^\top c$, so
$u_k^\top c=0$: the form vanishes whenever either argument lies in
$\Lambda_a$ (this covers both the antisymmetric block and the cross
terms), and on $\Lambda_s$ it is the stated rank-one square. The block
identities follow since $N_s^\top M_s(y)N_s$ (resp.\
$N_a^\top M_a(y)N_a$) is precisely the matrix of the restriction of
$\Gamma(y)$ to $\Lambda_s$ (resp.\ $\Lambda_a$) in the chosen basis: the
congruence preserves the pairing of lattice vectors.
\end{proof}

All claims of Lemma~\ref{lem:bridge} were additionally machine-verified in
exact integer arithmetic at $k=2,\dots,5$
(\texttt{verify\_bridge\_lemma.py}: the permutation labels and
involutivity, $\Pi u_k=u_k$, vanishing of the antisymmetric and cross
compressions of $\Gamma(y^\ast)$, and the rank-one symmetric compression
$\tilde u\tilde u^\top$), independently of the certificate bases; the
certificate-basis instances at $k=2,3,4$ are re-checked in the independent
re-verification (audit S3).

The gain vector lives in the symmetric kernel; let $K_s$ be a basis of
$\tilde u^\perp$ there. The \emph{second-order tangent program} at the
face is: over pinned directions ($y_1,y_2$ vanishing at the identity class
and the $k$ face-coordinate classes) with $(y_1,\lam)\in V_k$, $\lam\ge0$,
\begin{equation}\label{eq:tangent}
\text{maximize } B_0\cdot y_2-m\cdot y_1
\quad\text{s.t.}\quad
\begin{pmatrix}M_{s0}&M_s(y_1)N_sK_s\\ \ast&(N_sK_s)^\top
M_s(y_2)N_sK_s\end{pmatrix}\succeq0,
\end{equation}
plus the analogous antisymmetric block
$\bigl(\begin{smallmatrix}M_{a0}&M_a(y_1)N_a\\ \ast&N_a^\top
M_a(y_2)N_a\end{smallmatrix}\bigr)\succeq0$. Its value $v_k$ lower-bounds
$a_k$ once the framework step of Section~\ref{sec:arc} is in place. Any
feasible point gives $v_k\ge$ its objective; by
Lemma~\ref{lem:identities} the objective of a feasible point is
$B_0\cdot y_2+\lam/8$.

\begin{lemma}[range identity]\label{lem:range}
On $V_k$, for every $k$, both block range conditions hold identically:
$\range\bigl(M_s(y_1)N_sK_s\bigr)\subseteq\range M_{s0}$ and
$\range\bigl(M_a(y_1)N_a\bigr)\subseteq\range M_{a0}$.
\end{lemma}

\begin{proof}
$M_{s0}$ is PSD with $\ker M_{s0}=\Span N_s$ exactly
(Theorem~\ref{thm:L2}), so $\range M_{s0}=(\Span N_s)^\perp$ and the
condition reads $N_s^\top M_s(y_1)N_sK_s=0$. On $V_k$,
$N_s^\top M_s(y_1)N_s=\lam\tilde u\tilde u^\top$ and $K_s$ is by
definition a basis of $\tilde u^\perp$, so
$N_s^\top M_s(y_1)N_sK_s=\lam\tilde u(\tilde u^\top K_s)=0$. The
antisymmetric condition is $N_a^\top M_a(y_1)N_a=0$, part of the
definition of $V_k$.
\end{proof}

Hence, by Albert's generalized Schur criterion, block-PSD is equivalent to
Schur-complement-PSD at every point of $V_k$, for every $k$ --- the range
identity observed per level in~\cite{overshoot} is structural.

\begin{lemma}[pin homogeneity and objective constants]\label{lem:pins}
For every $k$:
\begin{enumerate}
\item All pins of the tangent program are homogeneous ($y_1$ and $y_2$
  vanish at the identity class and the $k$ face-coordinate classes), so
  $V_k$ is a linear space and $t\,y_2$ is pin-feasible for all $t$ if
  $y_2$ is.
\item The uniform Slater point $y_2^\ast$ (Theorem~\ref{thm:L4}) satisfies
  all pins, and $B_0\cdot y_2^\ast=B_0\cdot e_{\mathrm{class}(1)}
  -B_0\cdot y_0(0)=0-4=-4$ exactly ($B_0$ has no identity-class
  coefficient; $B_0\cdot y\equiv4$ on the affine face hull, including at
  $\delta=0$, by Theorem~\ref{thm:L1}(4)).
\item On $V_k$: $(y_1)_{\la A_0\ra}=-\lam/8$ and $B_0\cdot y_1=0$
  (Lemma~\ref{lem:identities}), so the objective of~\eqref{eq:tangent} is
  $B_0\cdot y_2+\lam/8$.
\end{enumerate}
\end{lemma}

\subsection{The feasible point}

\begin{theorem}[small-$\lambda$ positivity]\label{thm:smalllambda}
Fix $k\ge2$ and suppose a witness exists at level $k$: a pinned direction
$y_1(w)$ with
\[
N_s^\top M_s(y_1(w))N_s=\tilde u\tilde u^\top,\qquad
N_a^\top M_a(y_1(w))N_a=0
\]
(i.e.\ $(y_1(w),1)\in V_k$). Let
\[
C_s=(N_sK_s)^\top M_s(y_2^\ast)N_sK_s,\qquad
C_a=N_a^\top M_a(y_2^\ast)N_a\qquad(\text{both}\succ0\text{ by
Theorem~\ref{thm:L4}}),
\]
\[
Q_s=F^\top M_{s0}^{+}F,\ \ F=M_s(y_1(w))N_sK_s;\qquad
Q_a=G^\top M_{a0}^{+}G,\ \ G=M_a(y_1(w))N_a,
\]
and $\rho_k(w)=\min\{t:\;tC_s\succeq Q_s,\ tC_a\succeq Q_a\}$ (finite
since $C_s,C_a\succ0$; nonnegative since $Q_s,Q_a\succeq0$). Then for
every rational $t_0>\rho_k(w)$,
\[
v_k\;\ge\;\frac1{1024\,t_0}\;>\;0 .
\]
\end{theorem}

\begin{proof}
Take any $\lam\in\bigl(0,\tfrac1{32t_0}\bigr)$ and the explicit point
\[
y_1=\lam\,y_1(w),\qquad y_2=(t_0\lam^2)\,y_2^\ast .
\]
Pins hold by Lemma~\ref{lem:pins} (homogeneity); $\lam\ge0$;
$(y_1,\lam)\in V_k$ by linearity of $V_k$. Both block matrices
in~\eqref{eq:tangent} are PSD by Albert's criterion: the corner blocks
$M_{s0},M_{a0}\succeq0$ (Theorem~\ref{thm:L2}), the range conditions hold
(Lemma~\ref{lem:range}), and the Schur complements are
\[
(t_0\lam^2)\,C_s-\lam^2Q_s=\lam^2\,(t_0C_s-Q_s)\succeq0,
\]
and the same in the antisymmetric block. The objective is
$B_0\cdot y_2+\lam/8=\lam/8-4t_0\lam^2$ (Lemma~\ref{lem:pins}), maximized
over the allowed $\lam$ at $\lam=\tfrac1{64t_0}$, with value
$\tfrac1{1024t_0}>0$.
\end{proof}

By Theorem~\ref{thm:witness} the witness exists at every level:
$w=y^\ast/4$ is pinned and $P(y^\ast/4)=\tfrac14u_ku_k^\top$, so
Lemma~\ref{lem:bridge}(iii) --- applied with $\lam=\tfrac14$, i.e.\
$\tilde u=u_k/2$ in the certificate scaling --- delivers exactly the two
swap-split hypotheses $N_s^\top M_s(w)N_s=\tilde u\tilde u^\top$ and
$N_a^\top M_a(w)N_a=0$, for every $k\ge2$ at once. (The same chain is
verified exactly in the concrete certificate bases at $k=2,3,4$ in the
independent re-verification.) Hence $v_k>0$ \emph{for every} $k\ge2$.

\section{Exact arc feasibility: from tangent jet to $c_k(s)$}
\label{sec:arc}

The step from a feasible tangent jet to the true SDP value is where naive
second-order arguments fail: when the first-order compressed form
degenerates on the kernel, the standard $o(s^2)$-correction argument
breaks down. Our jets degenerate on $\tilde u^\perp$ --- but exactly there
the second-order block is strictly positive definite, and the rank-one
first-order form has exactly zero kernel cross-coupling. That structure
makes any correction unnecessary: the arc is exactly feasible as written.

\begin{lemma}[no-repair positivity]\label{lem:norepair}
Let $M(s)=M_0+sM_1+s^2M_2$ be symmetric $n\times n$ and decompose
$\mathbb R^n=R\oplus U\oplus W$ orthogonally with $\dim U\le1$. Assume
\begin{enumerate}
\item[(a)] $M_0\succeq0$ with $\ker M_0=U\oplus W$ exactly;
  $c:=\min\operatorname{spec}(M_0|_R)>0$;
\item[(b)] $P_{\ker}M_1P_{\ker}=\hat\lam\,P_U$ with $\hat\lam>0$
  (rank one, aligned with $U$; if $U=0$, the kernel-compressed first-order
  form vanishes). In particular the $U$--$W$ and $W$--$W$ first-order
  blocks are exactly zero;
\item[(c)] no condition on the range--kernel coupling
  $X_1:=P_RM_1P_{\ker}$;
\item[(d)] $H_{WW}:=P_WM_2P_W-X_{1W}^\top M_0^{+}X_{1W}\succeq\mu I_W$
  with $\mu>0$ ($X_{1W}=P_RM_1P_W$).
\end{enumerate}
Set $b_1=\|M_1\|$, $b_2=\|M_2\|$ (spectral norms), $h=\|H_{UU}\|$,
$\kappa=\|H_{UW}\|$ (blocks of
$H:=P_{\ker}M_2P_{\ker}-X_1^\top M_0^{+}X_1$; note
$h,\kappa\le b_2+b_1^2/c$), and
\[
\eps:=\frac{2b_1^2(b_1+b_2)}{c^2}+\frac{4b_1b_2}{c}+\frac{2b_2^2}{c},
\qquad
s^\ast:=\min\Bigl\{1,\ \frac c{2(b_1+b_2)},\ \frac{\hat\lam}{2(h+\eps)},\
\frac\mu{2\eps},\ \frac{\hat\lam\mu}{4(\kappa+\eps)^2}\Bigr\}
\]
(dropping the $\hat\lam$-terms if $U=0$). Then $M(s)\succeq0$ for all
$s\in(0,s^\ast]$ (and $\succ0$ on the range-block--kernel-complement in
the interior). No $o(s^2)$ correction is needed.
\end{lemma}

\begin{proof}
For $s\le\min\{1,\tfrac c{2(b_1+b_2)}\}$:
$\|A(s)-A_0\|\le sb_1+s^2b_2\le s(b_1+b_2)\le c/2$, so the range block
$A(s):=P_RM(s)P_R\succeq\tfrac c2I\succ0$ and $\|A(s)^{-1}\|\le2/c$. By
the Schur criterion, $M(s)\succeq0$ iff the kernel Schur complement
$S(s):=K(s)-X(s)^\top A(s)^{-1}X(s)\succeq0$, where
$K(s)=s\hat\lam P_U+s^2G_2$ ($G_2=P_{\ker}M_2P_{\ker}$; the $s$-term off
$U$ vanishes \emph{exactly} by (b), the decisive consequence of the
rank-one form) and $X(s)=sX_1+s^2X_2$. Expanding around
$A_0^{-1}=M_0^{+}|_R$:
\[
X(s)^\top A(s)^{-1}X(s)=s^2X_1^\top M_0^{+}X_1+\mathrm{Rem}(s),
\]
\[
\|\mathrm{Rem}(s)\|\le s^2b_1^2\cdot\frac{2s(b_1+b_2)}{c^2}
+2s^3b_1b_2\cdot\frac2c+s^4b_2^2\cdot\frac2c\;\le\;s^3\eps\qquad(s\le1),
\]
using $\|A(s)^{-1}-A_0^{-1}\|\le\|A_0^{-1}\|\,\|A(s)-A_0\|\,\|A(s)^{-1}\|
\le2s(b_1+b_2)/c^2$. Hence $S(s)=s\hat\lam P_U+s^2H-\mathrm{Rem}(s)$ with
the block bounds
\[
\text{$U$-entry }\alpha\ge s\hat\lam-s^2h-s^3\eps\ge\tfrac{s\hat\lam}2
\quad\text{if }s\le\tfrac{\hat\lam}{2(h+\eps)};
\]
\[
\text{$W$-block }D\succeq s^2\mu I_W-s^3\eps I_W
\succeq\tfrac{s^2\mu}2I_W\quad\text{if }s\le\tfrac\mu{2\eps};
\qquad
\text{cross }\|r\|\le s^2\kappa+s^3\eps\le s^2(\kappa+\eps).
\]
For the $2\times2$ block form
$\bigl(\begin{smallmatrix}\alpha&r^\top\\ r&D\end{smallmatrix}\bigr)$
with $D\succ0$, PSD holds iff $\alpha\ge r^\top D^{-1}r$, for which
$\alpha\cdot\lambda_{\min}(D)\ge\|r\|^2$ suffices:
$\tfrac{s\hat\lam}2\cdot\tfrac{s^2\mu}2\ge s^4(\kappa+\eps)^2$ iff
$s\le\hat\lam\mu/(4(\kappa+\eps)^2)$. If $U=0$ then
$S(s)\succeq s^2(\mu-s\eps)I\succeq0$ directly.
\end{proof}

\begin{remark}[necessity of the strict margin]\label{rem:mu0}
If $\mu=0$ --- the $W$-margin merely PSD, as happens at \emph{exactly
optimal} jets by complementary slackness --- the $s^3$ remainder can go
negative on $W$ and a genuine third-order repair is required; this is
precisely the measured third-order infeasibility of the raw certified
ladder jets of~\cite{overshoot} (their blends are minute --- blend weight
$\eps_p=10^{-54}$ at $k=2$, $10^{-7}$ at $k=3,4$ --- so the repaired
margins $\mu$, and with them $s^\ast$, are minuscule, and violations are
visible at $s\sim0.05$).
Hypothesis $\mu>0$ is the sharp boundary: the framework caveat flagged
in~\cite{overshoot} was the $\mu=0$ case, not a hole in the strict case.
\end{remark}

\begin{theorem}[exact arc; witness $\Rightarrow$ $a_k>0$]\label{thm:arc}
Fix $k\ge2$ and a witness $w$ as in Theorem~\ref{thm:smalllambda}. Let
$(y_1,y_2)=(\lam\,y_1(w),\,t_0\lam^2y_2^\ast)$ with $t_0>\rho_k(w)$
rational and $\lam=\tfrac1{64t_0}$, and set
\[
y(s)\;=\;y_0(\delta^\ast)+s\,y_1+s^2\,y_2 .
\]
Then there is an explicit $s^\ast_k>0$ (Lemma~\ref{lem:norepair}'s
constants, computed per swap block) such that for all $s\in(0,s^\ast_k]$:
\begin{enumerate}
\item $y(s)_1=1$ and $\Gamma_k(y(s))\succeq0$: the arc is \emph{exactly}
  feasible for the level-$k$ SDP at deformation $s$;
\item $B_s\cdot y(s)=4-s+g_ks^2+(t_0\lam^2)s^3$ \emph{exactly} (a cubic,
  not an expansion), with
  \[
  g_k=\frac\lam8-4t_0\lam^2=\frac1{1024\,t_0}>0 .
  \]
\end{enumerate}
Hence $c_k(s)\ge4-s+g_ks^2$ on $(0,s^\ast_k]$ (the cubic term is
positive), so $a_k\ge g_k>0$.
\end{theorem}

\begin{proof}
Swap-invariance of $y(s)$ gives $[\Gamma(y(s)),S]=0$, so the congruence
block-diagonalizes: $\Gamma\succeq0\iff M_s(s)\succeq0$ and
$M_a(s)\succeq0$. Apply Lemma~\ref{lem:norepair} to each block.
\emph{Symmetric block:} $R=\range M_{s0}$, $\ker M_{s0}=\Span N_s$
exactly (Theorem~\ref{thm:L2}); $W=\range(N_sK_s)$, $U$ its
orthocomplement in the kernel ($1$-dimensional: the kernel quadratic form
is $q_1(N_sv)=\lam(\tilde u^\top v)^2$ and $N_s$ is injective, so
$\hat\lam=\lam(\tilde u^\top v_U)^2>0$). Hypothesis (b) holds because
$N_s^\top M_s(y_1)N_s=\lam\tilde u\tilde u^\top$ makes the kernel
bilinear form $b_1(N_sv,N_sv')=\lam(\tilde u^\top v)(\tilde u^\top v')$,
which vanishes whenever either argument lies in $W$. Hypothesis (d) is
the strict $t_0$-margin: in $N_sK_s$-coordinates
\[
(N_sK_s)^\top M_s(y_2)N_sK_s-Q_s(y_1)=\lam^2\,(t_0C_s-Q_s)\succeq
\lam^2m\,I\qquad(m=\lambda_{\min}>0,\text{ certified exactly}),
\]
and the orthonormal-frame congruence $E_W=(N_sK_s)T$ rescales the margin
by $\sigma_{\min}(T)^2=1/\|N_sK_s\|^2>0$. \emph{Antisymmetric block:} the
$U=0$ case ($N_a^\top M_a(y_1)N_a=0$ on $V_k$), margin
$\lam^2(t_0C_a-Q_a)\succ0$. Feasibility (1) follows on $(0,s^\ast_k]$
with $s^\ast_k$ the minimum of the two blocks' $s^\ast$; the pins give
$y(s)_1=1+0+0$.

Objective (2): $B_s=B_0-s\,m\cdot$, and
\[
B_s\cdot y(s)=B_0\cdot y_0+s(B_0\cdot y_1-m\cdot y_0)
+s^2(B_0\cdot y_2-m\cdot y_1)+s^3(-m\cdot y_2)
=4-s+g_ks^2+t_0\lam^2s^3,
\]
using $B_0\cdot y_0=4$, $m\cdot y_0=1$, $B_0\cdot y_1=0$,
$m\cdot y_1=-\lam/8$, $B_0\cdot y_2=-4t_0\lam^2$, and
$m\cdot y_2=-t_0\lam^2$ (Lemma~\ref{lem:pins} together with
$(y_2^\ast)_{\la A_0\ra}=-1$). Feasibility gives
$c_k(s)\ge B_s\cdot y(s)$, and the cubic term is positive.
\end{proof}

Two machine checks accompany the lemma ($k=2,3,4$).
\texttt{verify\_framework\_lemma.py} re-verifies the structural hypotheses
of Lemma~\ref{lem:norepair} exactly (base kernels, rank-one alignment,
strict $t_0$-margins), verifies $M_s(s),M_a(s)$ exactly positive definite
at $s=10^{-3},10^{-4}$ by rational pivots with the objective cubic exact,
and reports the lemma constants and $s^\ast_k$ in floating point.
\texttt{verify\_sstar\_rational.py} then certifies the constants
themselves in exact rational arithmetic --- lower bounds on $c$ and $\mu$
by rational PD certificates against exact Gram matrices, $\hat\lam$ and
$h$ as exact rationals (the $U$-direction is rational), upper bounds on
$b_1,b_2,\kappa$ by exact PD and square comparisons --- and outputs
certified rational interval endpoints
$s^\ast_2\ge1.6\times10^{-3}$, $s^\ast_3\ge9.2\times10^{-3}$,
$s^\ast_4\ge7.5\times10^{-3}$ (the exact fractions are printed by the
verifier), each $\ge10^{-3}$, so both exact-PD sample points lie inside
the certified intervals. The analysis is
quantitatively tight, not just a bound: the measured minimal eigenvalue at
$k=2$, $s=10^{-3}$ is $+2.74\times10^{-10}$, which matches $s^2\mu$ of the
symmetric block; the lemma's guaranteed floor is the halved, min-over-blocks
value $s^2\mu/2\approx1.5\times10^{-11}$, comfortably below it.

\begin{corollary}[the certified ladder bounds transfer to $a_k$]
\label{cor:ladder}
Lemma~\ref{lem:norepair} applies verbatim to any strictly feasible point
of the tangent program (rank-one first-order form, exact base kernel,
strictly PD Schur blocks): its objective value lower-bounds $a_k$. The
certified ladder jets of~\cite{overshoot}
(\texttt{verify\_v\{2,3,4\}\_rational\_base.py}) are blends with strictly
positive-definite Schur complements (Lemma~\ref{lem:blend}), rank-one
first-order forms and the exact base kernel, so their certified values
transfer unconditionally:
\[
a_2>\tfrac1{39},\qquad a_3>\tfrac1{188},\qquad a_4>\tfrac1{641}.
\]
The framework caveat of~\cite{overshoot} is thereby discharged at
$k=2,3,4$ (the raw jets' $s^\ast$ is tiny because the blend weights are
$\eps_p=10^{-54}$ at $k=2$ and $10^{-7}$ at $k=3,4$, hence so are the
repaired Schur margins, which is why they \emph{look} third-order
infeasible at $s\sim0.05$; the lemma needs only some interval
$(0,s^\ast]$, and the liminf does the rest).
\end{corollary}

\section{Proof of the Main Theorem}\label{sec:assembly}

\begin{proof}[Proof of Theorem~\ref{thm:main}]
\emph{Lower interval bound, $k\ge2$.} By Theorem~\ref{thm:witness} the
witness $w=y^\ast/4$ exists at every level (pinned, certificate-scaled,
antisymmetric compression zero). By Theorem~\ref{thm:L4}, $C_s,C_a\succ0$
for all $k$, so $\rho_k(w)<\infty$ and a rational $t_0=t_0(k)>\rho_k(w)$
exists. Theorem~\ref{thm:arc} then gives explicit $s^\ast_k>0$ and
$g_k=1/(1024\,t_0(k))>0$ with $c_k(s)\ge4-s+g_ks^2$ on $(0,s^\ast_k]$,
hence $a_k\ge g_k>0$.

\emph{Non-exactness.} By~\eqref{eq:cQ} there is $s_1>0$ with
$c_Q(s)\le4-s+\tfrac{s^3}3$ on $(0,s_1]$. For
$\eps_k:=\min\{s^\ast_k,\,s_1,\,2g_k\}$ and $s\in(0,\eps_k]$,
\[
c_k(s)-c_Q(s)\;\ge\;g_ks^2-\tfrac{s^3}3\;\ge\;g_ks^2\bigl(1-\tfrac23\bigr)
\;>\;0 .
\]
For the remaining finite levels: the level-$1$ and almost-quantum
relaxations are outer to level $2$
($\mathcal W_1\subset\mathcal W_{1+AB}\subset\mathcal W_2$, the latter
inclusion of monomial sets verified explicitly in the independent
re-verification), so
$c_1(s)\ge c_{1+AB}(s)\ge c_2(s)>c_Q(s)$ on $(0,\eps_2]$. Hence
\emph{every} finite level fails to be exact on an explicit interval
$(0,\eps_k]$.

\emph{The exponent.} The lower bound $c_k(s)-(4-s)\ge g_ks^2$ on
$(0,s^\ast_k]$ is unconditional. For the matching upper bound we invoke the
almost-quantum expansion $c_{1+AB}(s)=4-s+\tfrac3{64}s^2+o(s^2)$
of~\cite{overshoot}; its $\le$ half rests on the dual second-order arc
constructed there, whose tangent conditions sit at the boundary of the PSD
cone (the $\mu=0$ regime of Remark~\ref{rem:mu0}), so we state the upper
half as conditional on that expansion. Monotonicity of the hierarchy then
gives
\[
g_ks^2\;\le\;c_k(s)-(4-s)\;\le\;c_2(s)-(4-s)\;\le\;
c_{1+AB}(s)-(4-s)\;=\;\tfrac3{64}s^2+o(s^2),
\]
the first inequality on $(0,s^\ast_k]$; so, granted the companion's
expansion, the level-$k$ overshoot is $\Theta(s^2)$ with
$g_k\le a_k\le\tfrac3{64}$. (None of the non-exactness statements above
depend on this step.)

\emph{Quantitative instances.} Corollary~\ref{cor:ladder} gives the
unconditional $a_2>\tfrac1{39}$, $a_3>\tfrac1{188}$, $a_4>\tfrac1{641}$.
\end{proof}

\begin{remark}[explicit constants]\label{rem:constants}
The construction is fully effective. With the certified per-level
witnesses of the ladder, $t_0(k)=0.316,\,29.7,\,302$ at $k=2,3,4$, giving
$g_2\ge\tfrac{15625}{5053008}\approx3.09\times10^{-3}$,
$g_3\approx3.29\times10^{-5}$, $g_4\approx3.23\times10^{-6}$; with the
universal witness $y^\ast/4$ itself the independent re-verification gave better constants at
higher levels, $t_0=0.348,\,8.59,\,69.4$. The interval endpoints at $k=2,3,4$ are
certified \emph{rational}: $s^\ast_2\ge1.6\times10^{-3}$,
$s^\ast_3\ge9.2\times10^{-3}$, $s^\ast_4\ge7.5\times10^{-3}$
(\texttt{verify\_sstar\_rational.py}, exact arithmetic throughout; the
floating-point evaluation of the lemma formula gives
$1.8\times10^{-3},\,1.0\times10^{-2},\,8.4\times10^{-3}$). At a general
level $k$, $s^\ast_k$ may be taken to be any positive rational below the
lemma's minimum; since all constants are eigenvalue/norm data of explicit
rational matrices, the same rational-PD certification procedure applies
verbatim at every $k$. The witness-dependent growth of $t_0(k)$ means the
theorem provides \emph{positivity uniform in kind, not in size}: no decay
law for $a_k$ is claimed (Section~\ref{sec:discussion}).
\end{remark}

\begin{remark}[reading of $a_k$]
$a_k$ is defined as a liminf; Theorem~\ref{thm:arc} in fact proves the
stronger \emph{interval} statement $c_k(s)\ge4-s+g_ks^2$ pointwise on
$(0,s^\ast_k]$, so no expansion or analyticity assumption on $c_k$ enters
the non-exactness conclusion. Whenever $c_k$ admits a second-order
expansion at $0^+$, its coefficient is $\ge g_k$.
\end{remark}

\section{The mechanism: the overshoot is non-quantum}\label{sec:mechanism}

Why did the witness have to be the signed formula~\eqref{eq:ystar}? The
following two no-gos locate the level-$k$ overshoot strictly outside
quantum theory's tangent directions; together they are the conceptual
heart of the result.

\begin{proposition}[GNS no-go: no state witness]\label{prop:gns}
For $k\ge3$, no witness $y_1$ with $P(y_1)=u_ku_k^\top$ can be realized as
the moment functional $\la\psi_1|\cdot|\psi_1\ra$ of a vector in any
Hilbert-space representation of the scenario algebra (commuting
involutions $A_i,B_j$).
\end{proposition}

\begin{proof}
If $y_1$ consists of moments of $\psi_1$, then $P(y_1)$ is the Gram matrix
of the kernel-defect vectors $K_a\psi_1$, and rank-one positivity in a
definite space forces $K_a\psi_1=u_a v$ with $|v|=1$. The
$u_k$-amplitudes then force: (1) $(A_0-1)\psi_1=(B_0-1)\psi_1=v$
(from the trailing-$0$ reductions with $|u|=1$); (2) $A_iv=B_jv=-v$ (from
the length-$3$ reductions, present once $k\ge3$); (3) $R_1\psi_1=0$
($u[D_\emptyset]=0$). Self-adjointness gives
$\la v,R_1\psi_1\ra=\la R_1v,\psi_1\ra=4\la v,\psi_1\ra=0$ and
$\la v,(A_0-1)\psi_1\ra=\la(A_0-1)v,\psi_1\ra=-2\la v,\psi_1\ra=0$, i.e.\
$|v|^2=\la v,(A_0-1)\psi_1\ra=0$: contradiction.
\end{proof}

\begin{remark}[indefinite (Krein) representations]\label{rem:krein}
The proposition is stated for Hilbert-space states only, and deliberately
so. In a Krein space the argument changes at its first step: an
indefinite rank-one Gram $[K_a\psi_1,K_b\psi_1]=u_au_b$ does \emph{not}
force the defect vectors to be collinear --- each $K_a\psi_1=u_av+n_a$
may be perturbed by mutually orthogonal \emph{neutral} vectors $n_a$
($[n_a,n_b]=0$, $[v,n_a]=0$) without changing the Gram --- so the chain
(1)--(3) need not start. What the same computation does show is
conditional: \emph{if} the defects are collinear, $K_a\psi_1=u_av$, then
self-adjointness forces $[v,v]=0$, i.e.\ the gain-carrying direction must
be neutral, which is impossible in a definite space but not in an
indefinite one. We therefore make no claim about general Krein-space
realizations. (Nothing in the main theorem depends on this section; the
propositions serve to locate the witness outside quantum tangent
directions.)
\end{remark}

\begin{proposition}[curve no-go: no smooth quantum family]\label{prop:curve}
For any smooth curve of genuine commuting-party quantum models
$(A(t),B(t),\psi(t))$ through the base realization, the compressed
first-order form cannot equal $\lam u_ku_k^\top$ with $\lam>0$ at any
order in $t$, for any $k\ge3$.
\end{proposition}

\begin{proof}
The dressed defect vectors obey
$\chi_{D_W}(t)=A_{w_a}(t)B_{w_b}(t)\,\chi_{D_\emptyset}(t)$ with unitary
prefactors, so $\|\chi_{D_W}(t)\|=\|\chi_{D_\emptyset}(t)\|$ for all $t$:
the compressed form has \emph{equal} diagonal at all dressings at every
order in $t$. But the gain direction requires $u[D_W]^2=4$ against
$u[D_\emptyset]^2=0$.
\end{proof}

Thus the direction that produces the quadratic overshoot is invisible not
only to states (Proposition~\ref{prop:gns}) but to all smooth families of
quantum models (Proposition~\ref{prop:curve}): \emph{the overshoot lives
strictly in the NPA-feasible, non-quantum part of the tangent cone at the
critical face}. Finite NPA levels fail near the critical line because
they admit a signed tangent direction that quantum mechanics forbids ---
and one level-independent signed quadratic ($y^\ast$) realizes it at
every level simultaneously, in the coherent-tower form predicted by the
germ stabilization law (Remark~\ref{rem:germ}).

\section{Discussion}\label{sec:discussion}

\subsection{What remains open}

\emph{The decay law.} The theorem gives $a_k\ge1/(1024\,t_0(k))>0$ with
$t_0(k)$ witness-dependent and growing ($0.35,\,8.6,\,69$ along
$y^\ast$); the measured coefficients $a_2\approx0.0281$,
$a_3\approx0.0096$ \cite{overshoot} decay with irregular successive
ratios ($\approx0.60$, then $\approx0.34$). Nothing here controls the decay; a uniform-in-$k$ lower
bound on $\min_w\rho_k(w)^{-1}$, or a level-to-level inequality
$a_{k+1}\ge r'a_k$, would sharpen positivity into a rate.

\emph{The dimension law (R2).} The first-order solution space $V_k$ has
$\dim V_k=2k+3$ verified at $k\le5$, with a constructive lower bound for
all $k$; the matching upper bound is open. It is \emph{not used}
anywhere in the proof --- the certificates never invoke it --- but it is
the one remaining structural conjecture of this framework.

\emph{The fixed-$s$ question.} Our theorem is local: for each $k$,
non-exactness on $(0,\eps_k]$ with $\eps_k\downarrow0$ allowed. Since
each level \emph{is} exact for $s$ above its threshold, the pointwise
required level $D(s)$ is finite for every fixed $s>0$; what we prove is
$D(s)\to\infty$ as $s\to0^+$. The stronger claim that some fixed $s$
admits no exact level is false in spirit here and we do not make it.

\subsection{Placement}

The result sharpens the picture drawn in~\cite{overshoot} and closes its
conjecture. In the degree-bounded noncommutative Fej\'er--Riesz theorem
of Klep--Levenson--McCullough~\cite{KLM}, strict positivity at margin
$\eps$ yields a factorization at some finite one-party degree
$M'(\eps)$; our family has margin $\sim s^3$ and the theorem proves the
complementary necessity statement --- no bounded level (hence no bounded
word length) certifies uniformly as $s\to0$. In the finite-convergence
theory of polynomial optimization, Nie's theorem~\cite{Nie} (through
Marshall's boundary Hessian condition~\cite{Marshall}) makes finite
exactness \emph{generic}; the doubly-tilted critical point is exactly a
failure point of the boundary Hessian condition (cubic touch, contact
order three), and our theorem shows that there the generic conclusion
genuinely fails --- the first unconditional instance of this failure mode
for the NPA hierarchy at an attained, self-tested optimum, obtained by
construction rather than by hardness reduction~\cite{hardness}. The
commutative shadow of the problem cannot see this: on the compact torus
the degenerating family is SOS at base degree
(\cite{overshoot}, following~\cite{Dritschel,Scheiderer}), so the
divergence is carried by the noncommutative moment structure --- in the
present language, by the signed non-quantum tangent direction of
Section~\ref{sec:mechanism}. The scaling limit of the critical
degeneracy is a rescaled Motzkin form~\cite{overshoot}, but the
finite-level failure proved here is not a scalar-polynomial degree
phenomenon; it is a property of the restricted certificate class of the
hierarchy, localized now in an explicit direction. Finally, the
single-tilt functional, whose touch is quadratic, is finitely exact at
level $1{+}AB$ with the Bamps--Pironio certificate~\cite{BampsPironio}
--- the dichotomy between quadratic and cubic contact is complete: contact
order two is finitely certifiable, and at the contact-order-three point
no finite level suffices.

\appendix

\section{Verification programs}\label{app:verif}

The complete verification suite, together with both papers, is available at
\url{https://github.com/tohafrit/npa-nonexactness}; \texttt{run\_all\_checks.sh}
reproduces every machine check below in about forty seconds. Every claim
entering the proofs is either a finite hand-checkable
argument or an exact machine computation (integer or rational arithmetic
throughout; no floating point enters any proof). The finite computations
are constituent parts of the proofs, with the degree/coverage lemmas
(\ref{lem:regime}, \ref{lem:grid}) making them conclusive. The ancillary
files and what each one establishes:

\begin{center}
\small
\begin{tabular}{@{}>{\raggedright\arraybackslash}p{4.4cm}>{\raggedright\arraybackslash}p{7.8cm}l@{}}
\toprule
file & checks & runtime\\
\midrule
\texttt{verify\_pipeline\_laws.py} & counting laws; Lemma~\ref{lem:dual}
($E\succ0$, affine identity); Lemma~\ref{lem:closure} (rank $15/17$);
product-law identity $\Gamma(y_0(\delta))N_k=0$; affine dimension $k$;
Theorem~\ref{thm:L2} (exact PD of $G_k$, $k\le8$; $(k{+}1,k)$ splits);
$u_k$ facts and Lemma~\ref{lem:identities} on full solution spaces
($k\le5$); Theorem~\ref{thm:L4} ($K^\top\Gamma(y_2^\ast)K=K^\top K$,
$k\le5$) & $\sim$2\,s\\
\texttt{verify\_r1\_induction.py} & Theorem~\ref{thm:witness}: per-level
$P(y^\ast)=u_ku_k^\top$ at $k=2..8$; sharp pins; the regime-grid engine,
$E(K_1,K_2)=0$ for all $1{,}413{,}721$ pairs at side length $\le20$
(margin: $2{,}883{,}601$ at $\le24$); exact integers & $\sim$4\,s\\
\texttt{verify\_bridge\_lemma.py} & Lemma~\ref{lem:bridge}: swap
permutation of the $T/D$ basis, $\Pi u_k=u_k$, vanishing antisymmetric
and cross compressions, rank-one symmetric compression
$\tilde u\tilde u^\top$; exact integers, $k=2..5$, independent of the
certificate bases & $<$1\,s\\
\texttt{verify\_small\_lambda.py} & Theorem~\ref{thm:smalllambda}:
range identity, pins, $C_s,C_a\succ0$, $\rho_k$, the feasible point and
its value, $k=2,3,4$, exact & $<$1\,s\\
\texttt{verify\_framework\_}\newline\texttt{lemma.py} & Lemma~\ref{lem:norepair} and
Theorem~\ref{thm:arc}: structural hypotheses re-verified exactly; exact
rational PD of $M_s(s),M_a(s)$ at $s=10^{-3},10^{-4}$; exact cubic
objective; float report of the lemma constants; $k=2,3,4$ & $<$1\,s\\
\texttt{verify\_sstar\_}\newline\texttt{rational.py} & certified rational bounds
for every Lemma~\ref{lem:norepair} constant ($c,\mu$ lower, $b_1,b_2,
\kappa$ upper, $\hat\lam,h$ exact) and certified rational $s^\ast_k$
($\ge1.6\times10^{-3},\,9.2\times10^{-3},\,7.5\times10^{-3}$ at
$k=2,3,4$); exact rational arithmetic & $<$1\,s\\
\texttt{verify\_v\{2,3,4\}\_}\newline\texttt{rational\_base.py} & the certified ladder
jets behind Corollary~\ref{cor:ladder}: exact rational tangent-program
certificates for $v_2>\tfrac1{39}$, $v_3>\tfrac1{188}$,
$v_4>\tfrac1{641}$ (data in \texttt{v\{2,3,4\}\_cert\_data.py}) &
seconds--minutes\\
\texttt{audit\_spot.py} & the adversarial audit S1--S5 (per-item
provenance below), including the composition step: $y^\ast$ evaluated in
the certificate bases at $k=2,3,4$ gives exactly
$N_s^\top M_s(y^\ast)N_s=4uu^\top$, antisymmetric compression $0$, pins,
range identity & $\sim$11\,s\\
\texttt{scenario\_ext.py}, \texttt{npa\_general.py} & shared library:
level-$k$ word/class/moment-matrix machinery used by the verifiers &
---\\
\bottomrule
\end{tabular}
\end{center}

Each verifier exits $0$ iff all its checks pass; all pass. Supporting
exploration scripts (\texttt{attack\_r1.py},
\texttt{explore\_r1\_induction.py},
\texttt{r1\_formula\_hunt.py}) document, respectively, the dual criterion
and germ analysis of Remark~\ref{rem:germ} and the gauge in which
$y^\ast$ was found; they play no role in the proofs.

\emph{Independent re-verification.} The complete chain was checked a
second time, directed at the points where an error would be most damaging.
The machine checks of this audit are collected in \texttt{audit\_spot.py}
(sections S1--S5); the degree of independence varies by item, so we state
it per item. \emph{Freshly written code} (S1, S2: own word/reduction/%
kernel/$u_k$/$y^\ast$ implementations, sharing nothing with the
development stack): a reproduction of $P(y^\ast)=u_ku_k^\top$ at
$k=2,\dots,6$ with pins and anchors (S1), and an adversarial attack on the
fine-regime completeness of Lemmas~\ref{lem:regime}--\ref{lem:grid} with
$18{,}915$ targeted extreme pairs and $150{,}000$ random pairs at side
lengths up to $61$ (S2), backed by a hand re-derivation of
Lemma~\ref{lem:nest}(iii). \emph{Independent driver scripts over the
shared exact stack} (S3, S4 import \texttt{scenario\_ext} and the
\texttt{LevelData} machinery of \texttt{verify\_small\_lambda.py}): the
composition of Theorem~\ref{thm:witness} with
Theorem~\ref{thm:smalllambda} in the concrete certificate bases at
$k=2,3,4$ (the one step not covered by the original programs, now checked
exactly) together with the normalization bookkeeping of
Remark~\ref{rem:norm} (S3); and exact positive definiteness of the full
\emph{unsymmetrized} $\Gamma(y(s))$ at $s=10^{-3}$, bypassing the swap
splitting entirely --- run at $k=2$ and $3$ only (the swap-split exact PD
checks of \texttt{verify\_framework\_lemma.py} cover $k=2,3,4$) --- with
the objective cubic recomputed from the raw functional (S4). The audit
also re-derived Lemma~\ref{lem:norepair} by hand. Two further items
quoted above live elsewhere: the extended enumeration at side length
$\le26$ ($3{,}946{,}645$ pairs, all zero) is a rerun of the development
engine, reproducible via \texttt{check\_all\_pairs(26)} in
\texttt{verify\_r1\_induction.py}, and the base-point integrals of
Theorem~\ref{thm:L2} are checked exactly in
\texttt{verify\_pipeline\_laws.py} (section D). No discrepancy was found
in any step.


\begin{thebibliography}{99}
\bibitem{Gigena} N.~Gigena, E.~Panwar, G.~Scala, M.~Ara\'ujo, M.~Farkas,
A.~Chaturvedi, \emph{Self-testing tilted strategies for maximal
loophole-free nonlocality}, npj Quantum Inf.\ \textbf{11}, 82 (2025);
arXiv:2405.08743.
\bibitem{NPA} M.~Navascu\'es, S.~Pironio, A.~Ac\'in, \emph{A convergent
hierarchy of semidefinite programs characterizing the set of quantum
correlations}, New J.\ Phys.\ \textbf{10}, 073013 (2008).
\bibitem{overshoot} A.~Pakhunov, \emph{A phase transition in the exactness
of the NPA hierarchy at the critical doubly-tilted CHSH functional},
companion note (2026).
\bibitem{sep} A.~Pakhunov, \emph{A certified separation of the
almost-quantum set from the quantum set in the minimal Bell scenario},
companion note (2026).
\bibitem{BB} V.~Barizien, J.-D.~Bancal, \emph{Quantum statistics in the
minimal scenario}, Nat.\ Phys.\ \textbf{21}, 577 (2025); arXiv:2406.09350.
\bibitem{KLM} I.~Klep, J.~Levenson, S.~McCullough, \emph{Fej\'er--Riesz
factorization for positive noncommutative trigonometric polynomials},
arXiv:2511.09267 (2025).
\bibitem{Nie} J.~Nie, \emph{Optimality conditions and finite convergence
of Lasserre's hierarchy}, Math.\ Program.\ \textbf{146}, 97 (2014);
arXiv:1206.0319.
\bibitem{Marshall} M.~Marshall, \emph{Representations of non-negative
polynomials, degree bounds and applications to optimization}, Canad.\ J.\
Math.\ \textbf{61}, 205 (2009); \emph{Positive Polynomials and Sums of
Squares}, Math.\ Surveys Monogr.\ \textbf{146}, AMS (2008).
\bibitem{BampsPironio} C.~Bamps, S.~Pironio, \emph{Sum-of-squares
decompositions for a family of CHSH-like inequalities and their
application to self-testing}, Phys.\ Rev.\ A \textbf{91}, 052111 (2015);
arXiv:1504.06960.
\bibitem{hardness} L.~F.~Vargas, \emph{On the hardness of deciding the
finite convergence of Lasserre hierarchies}, Numer.\ Algebra Control
Optim.\ \textbf{16}, 35 (2026); arXiv:2401.12613.
\bibitem{Dritschel} M.~A.~Dritschel, \emph{Factoring non-negative operator
valued trigonometric polynomials in two variables}, Math.\ Ann.\
\textbf{391}, 515 (2024); arXiv:1811.06005.
\bibitem{Scheiderer} C.~Scheiderer, \emph{Sums of squares on real
algebraic surfaces}, Manuscripta Math.\ \textbf{119}, 395 (2006);
\emph{Non-existence of degree bounds for weighted sums of squares
representations}, J.\ Complexity \textbf{21}, 823 (2005).
\end{thebibliography}
\end{document}